\documentclass[12pt, draftclsnofoot, onecolumn]{IEEEtran}
\usepackage{graphicx,indentfirst}
\usepackage{booktabs}
\usepackage{indentfirst}
\usepackage{amsmath}
\usepackage{cite}
\usepackage{stfloats}
\usepackage{multicol}
\usepackage{multirow}
\usepackage{slashbox}
\usepackage{subfigure}
\usepackage[english]{babel}
\usepackage{amsthm}
\usepackage{amsmath}
\usepackage{multirow}
\usepackage{paralist}
\usepackage{algorithm, algorithmic}
\usepackage{amssymb}
\usepackage{amssymb}
\usepackage{mathrsfs}
\usepackage{amsfonts}
\usepackage{color}
\newtheorem{proposition}{Proposition}
\newtheorem{theorem}{Theorem}
\newtheorem{assumption}{Assumption}

\theoremstyle{plain}
\newtheorem{Sup}{Sub-problem}

\IEEEoverridecommandlockouts

\begin{document}
\bibliographystyle{IEEE2}
\title{Dynamic Pricing for Revenue Maximization\\ in Mobile Social Data Market \\with Network Effects}
\author{Zehui~Xiong, Dusit~Niyato,~\IEEEmembership{Fellow,~IEEE,} Ping~Wang,~\IEEEmembership{Senior Member,~IEEE,} Zhu~Han,~\IEEEmembership{Fellow,~IEEE,} and~Yang~Zhang,~\IEEEmembership{Member,~IEEE}\\
\thanks{Z. Xiong, D. Niyato and P. Wang are with School of Computer Science and Engineering, Nanyang Technological University, Singapore. Z. Han is with Electrical and Computer Engineering, University of Houston, Houston, USA. Y. Zhang is with School of Computer Science and Technology, Wuhan University of Technology, China.
\vspace*{-6mm}}}
\maketitle
\vspace*{-8mm}
\begin{abstract}
Mobile data demand is increasing tremendously in wireless social networks, and thus an efficient pricing scheme for social-enabled services is urgently needed. Though static pricing is dominant in the actual data market, price intuitively ought to be dynamically changed to yield greater revenue. The critical question is how to design the optimal dynamic pricing scheme, with prospects for maximizing the expected long-term revenue. In this paper, we study the sequential dynamic pricing scheme of a monopoly mobile network operator in the social data market. In the market, the operator, i.e., the seller, individually offers each mobile user, i.e., the buyer, a certain price in multiple time periods dynamically and repeatedly. The proposed scheme exploits the network effects in the mobile users' behaviors that boost the social data demand. Furthermore, due to limited radio resource, the impact of wireless network congestion is taken into account in the pricing scheme. Thereafter, we propose a modified sequential pricing policy in order to ensure social fairness among mobile users in terms of their individual utilities. We analytically demonstrate that the proposed sequential dynamic pricing scheme can help the operator gain greater revenue and mobile users achieve higher total utilities than those of the baseline static pricing scheme. To gain more insights, we further study a simultaneous dynamic pricing scheme in which the operator determines the pricing strategy at the beginning of each time period. Mobile users decide on their individual data demand in each time period simultaneously, considering the network effects in the social domain and the congestion effects in the network domain.  We construct the social graph using Erd\H{o}s-R\'enyi (ER) model and the real dataset based social network for performance evaluation. The numerical results corroborate that the dynamics of pricing schemes over static ones can significantly improve the revenue of the operator. 
\end{abstract}
\newpage

\section{Introduction}\label{Sec:Introduction}
The explosion of social application services on mobile platforms such as Facebook, Twitter, and WhatsApp leads to increasing demand of mobile social data. The mobile social services allow users to interact with each other online, and in turn the users are spending increasing amount of time on social service websites~\cite{backstrom2006group}. In 2018, the number of online social media users from mobile platforms has reached 2.958 billion, accounting for 57\% of mobile users~\cite{GlobalDigitalStatistics}. Reciprocally, when more users access more social services, they become more socially connected and their social ties are stronger, leading to even more social data consumption and interpersonal communication~\cite{easley2010networks}. This is verified in~\cite{GlobalDigitalStatistics} that over half of the cellular data consumption comes from the social media activities in mobile platforms and this percentage keeps growing in recent years. In a network, when a user increases its activity in a social service, its social friends are likely to increase their activities accordingly. This phenomenon that social data demand of one user is positively affected by the demand of other users is called \textit{network~effect} in economics~\cite{katz1994systems}. 
Therefore, as the mobile users pay the data fees to access the social services, the mobile network operator has an incentive to encourage more mobile users, i.e., its potential customers, to access the services by consuming more social data. Generally, the stronger the network effects are, the more revenue the mobile network operator gains~\cite{gong2015network, nie2017socially}.

The ``network effect'' is often highlighted in social and economic fields~\cite{tucker2008identifying}. For example, a few existing works~\cite{candogan2012optimal, hartline2008optimal} propose the pricing under equilibrium conditions and guide the operation of a social network concerning the network effects. Nevertheless, as one of the major issues in network economics~\cite{chiu2011exploring}, the network effects have also been investigated in communication networks such as Internet, mobile ad-hoc network (MONET) and peer-to-peer (P2P) networks~\cite{easley2010networks}. However, this potential benefit suffers from the limited capacity in physical communication networks, e.g., bandwidth. The reason is that when users increase their data demand, they may bear higher congestion, e.g., service delay, which restrains them to access and consume more. Consequently, the increasing congestion creates a significant barrier for the network operator, and thus leads to a lower revenue~\cite{gong2015network}. Therefore, mobile users' data demand is not only subject to the network effects in the social domain, but also the congestion effects in the network domain. However, this issue has been largely neglected by the studies in the literature.

In the revenue maximization, pricing plays a crucial role in the mobile data market because it directly affects the network operator's revenue and the user's utilities~\cite{sen2013survey}. Traditionally, the operator only adopted static pricing which is simple flat-rate data plans to engage their users. However, recently, with the popularity of mobile devices and rapid rise of online apps or videos, dynamic pricing emerges as an attractive alternative to yield better adaptation to unpredictable user data demand. The motivation is simple and intuitive: pricing should be leveraged strategically to affect demand to better handle unexploited capacity, and thus yield more revenue~\cite{hxuTCC}. Take some companies in the real world for example, MTN in Uganda and Uninor in India have offered time-dependent pricing for mobile message, where the message price is dynamically changed after one day or even one hour to achieve balance between the supply and demand~\cite{ha2012tube, sen2013survey}. Moreover, China Telecom charges its users a discount data fees during less congested times, e.g., at night, and normal data fees otherwise. Dynamic pricing has become an active field of the revenue management literature, with successful real-world applications in industries, e.g., cloud computing~\cite{hxuTCC}, smart grid~\cite{ferdous2017optimal} and spectrum trading~\cite{yang2013pricing}. With the flexibility to change the price, the optimal dynamic pricing policy for a mobile network operator towards maximizing the expected long-term revenue becomes appealing.

Nevertheless, the presence of network effects and congestion effects in the mobile social data market poses a remarkable challenge to the operation of dynamic pricing therein. To the best of our knowledge, this paper is the first to study the optimal dynamic pricing schemes of a mobile network operator, i.e., a seller, selling social data to a set of mobile users in social data market, where users' behavior is subject to both network effects and congestion effects. We assume that the mobile network operator, i.e., the seller has complete information about the social network and can perfectly charge each user differently, i.e., discriminatory pricing\footnote{The uniform pricing which is more popular is just a special case of the discriminatory pricing, and hence the techniques developed in the paper can be applied similarly and directly for uniform pricing.}~\cite{laffont1998network}. Particularly, the main contributions of this paper are summarized as follows:
\begin{itemize}
 \item Taking time-varying interactions between the mobile network operator and mobile users into account, we propose the sequential dynamic pricing scheme, where the operator individually offers each user a certain price in multiple time periods dynamically.
 \item We model the network effects in the social domain by utilizing the structural properties of the social network, which further increases the social data demand of mobile users. Furthermore, the model incorporates the congestion in the network domain to realistically capture the scarcity of radio resource in wireless network environment. Moreover, we analytically demonstrate that our proposed sequential dynamic pricing can help the mobile network operator to gain greater revenue and mobile users to achieve higher total utilities than those of existing optimal static pricing scheme.
 \item In addition, we propose a modified sequential pricing policy in order to ensure social fairness among mobile users with respect to their individual utilities. To gain more insights, the simultaneous dynamic pricing is developed in which the operator determines the pricing strategy at the beginning of each time period and users decide on their individual data demand in each time period simultaneously. We find the insights that the operator tends to offer the discount price to the users with more social influence which may bring more potential users subsequently.
 \item To characterize the network effects from social networks, we consider two social graphs. The first graph is constructed using the Erd\H{o}s-R\'enyi (ER) model~\cite{albert2000error}, and the second graph is built from a real dataset, i.e., the Brightkite dataset~\cite{cho2011friendship}. The performance evaluation corroborates the fact that the dynamics of pricing schemes over static ones can greatly improve the revenue of the operator.
\end{itemize}

The rest of the paper is organized as follows. In Section~\ref{Sec:Related}, we present a brief review on the related works. In Section~\ref{Sec:Model}, we formulate the revenue maximization problem in the dynamic model. In Section~\ref{Sec:Seqpricing}, we develop and analyze the sequential dynamic pricing to solve the optimization problem formulated in Section~\ref{Sec:Model}. Then, we propose and investigate the simultaneous pricing in Section~\ref{Sec:Simultaneous}. In Section~\ref{Sec:Simulation}, we compare the performance of the proposed dynamic pricing against existing optimal static pricing in terms of the revenue of the mobile network operator and total utilities of the mobile users. Additionally, we evaluate the role of different parameters on the proposed schemes. At last, we conclude the paper in Section~\ref{Sec:Conclusion}.

\section{Related Works}\label{Sec:Related}
The pricing schemes are designed to offer profitable business to the network operators, as well as to create favorable services for the users, which is a useful tool for addressing resource allocation issues~\cite{huang2013wireless}. Since users typically share limited communication resources in physical networks, e.g., bandwidth, congestion effects from physical domain on users' behavior is common and has been extensively studied for pricing strategies by a number of  works~\cite{mackie1995pricing, blumrosen2007welfare, sengupta2009economic, hande2010pricing}. For example, if an Internet network operator becomes oversubscribed, its subscribers suffer from the congestion externality due to limited bandwidth and radio resources. In~\cite{ma2016time}, the authors proposed an optimal design of time and location aware mobile data pricing, which motivates users to smooth traffic and reduce network congestion. Moreover, the network operators have been experiencing several emerging and innovative pricing schemes, e.g., rollover data plans, secondary market scheme~\cite{yu2017mobile} and sponsored data plan~\cite{xiong2017sponsor}. However, very few works take homophily phenomenon into consideration, i.e., network effects, for designing the pricing. The social aspect of mobile networking is an emerging paradigm for network design and optimization~\cite{gong2015network}. In~\cite{brown1987social}, the authors found that the information obtained from social tie connections will influence in decision making. Inspired by~\cite{brown1987social}, the social group utility maximization framework was proposed in~\cite{gong2017social, chen2016exploiting}, which captures the impact of users' diverse social ties that are subject to diverse social relationships. The authors in~\cite{blackburn2013last} showed the evidence of network effects in communication service using the real data analytic, and quantified such an effect using a simple metric. The finding of~\cite{blackburn2013last} provides the motivation for planning pricing mechanisms with network effects in the real markets.

Recently, the network effects have been considered jointly with service pricing from the economic perspective~\cite{hartline2008optimal}. For example, in the pioneering work~\cite{candogan2012optimal}, the authors investigated both the uniform pricing and the discriminatory pricing of the network operator in the presence of network effects. In~\cite{makhdoumi2017optimal, swapna2012dynamic}, the authors discussed the dynamic pricing strategy of divisible social goods with network effects. However, the authors investigated the user behaviors only in social networks without considering the congestion in wireless networks. Therefore, the model introduced therein has a major limitation as it cannot be applied into the wireless network environments, in which the radio resource is scarce and congestion can frequently happen. Subsequently, the authors in~\cite{gong2015network} formulated the interaction between a network operator and mobile users as a two-stage Stackelberg game, by extending the model presented in~\cite{candogan2012optimal}. The uniform pricing scheme was adopted, where the impacts of network effects and congestion were jointly analyzed. 
In~\cite{zhang2016social, zhang2017motivating}, the authors studied the similar problem to that in~\cite{gong2015network}, where the discriminatory pricing strategy was applied by the operator.
\begin{figure}[t]
\centering
\includegraphics[width=0.68\textwidth]{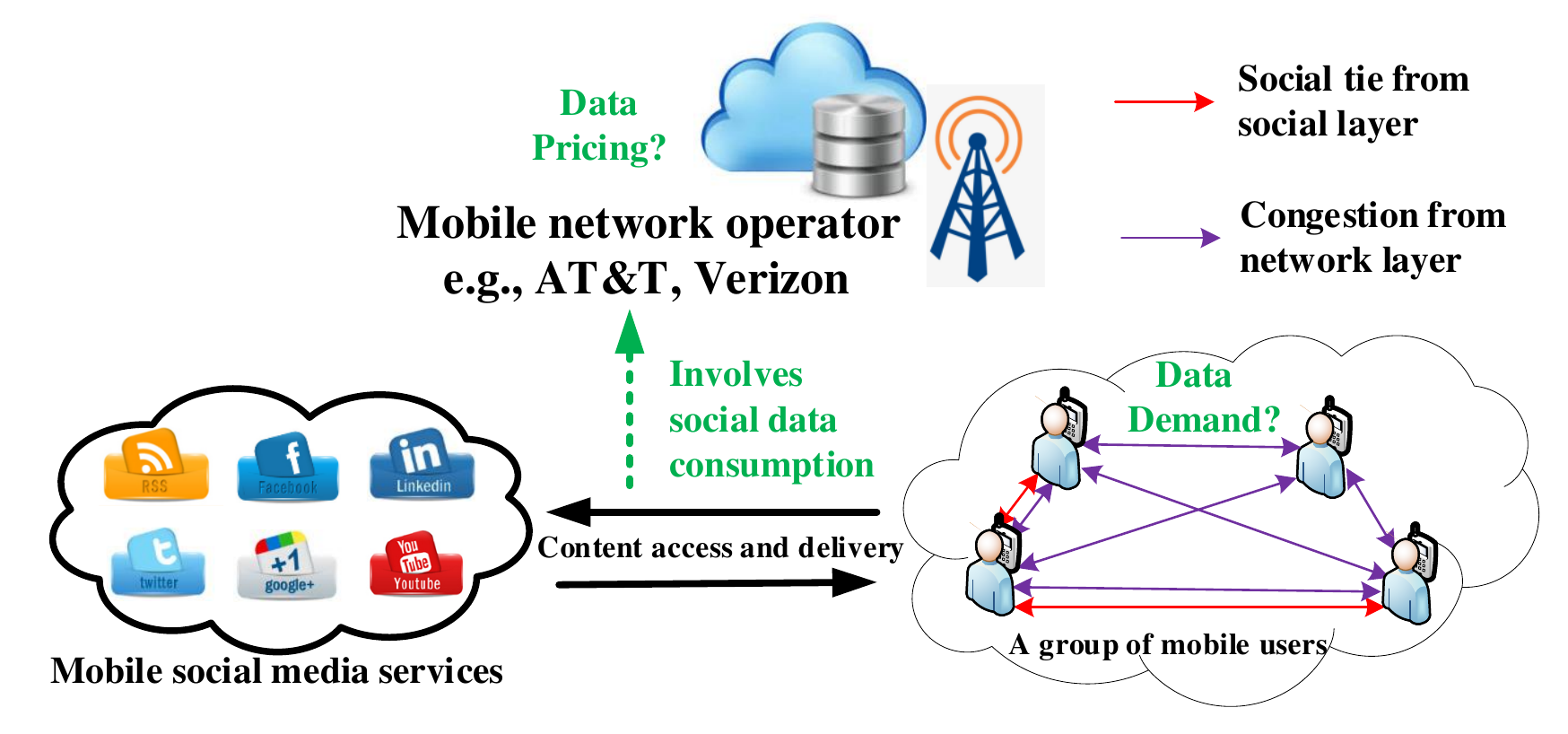}
\caption{System Model.}\label{Fig:Model}
\end{figure}
Nevertheless, they merely formulated the one-shot game to model the interaction between the network operator and mobile users with static pricing. In other words, the operator cannot utilize its ability to modify its strategy in response to the observed history. In view of this gap, we study the sequential dynamic pricing scheme in~\cite{xiong2017network}, where the users' behaviors are subject to both the network effects and congestion effects. In this paper, more analytical results have been added. Specifically, the sequential dynamic pricing scheme is implemented to address the social fairness issue. Moreover, we investigate the simultaneous dynamic pricing to gain more insights. 

\section{System Model}\label{Sec:Model}
\subsection{Basic static model}
In a social data market under our consideration, there is a set of mobile users ${\cal N} \buildrel \Delta \over = \{ 1, \ldots ,N\}$, as shown in Fig.~\ref{Fig:Model}.
Each mobile user $i \in {\cal N}$, i.e., the buyer, determines a non-negative quantity of the data demand from a Mobile Network Operator (MNO) for accessing social services, denoted by $x_i$ where ${x_i} \in [0,\infty )$. Let $\mathbf{x} \buildrel \Delta \over = ({x_1}, \ldots ,{x_N})$ denote the demand profile of all the users and ${\mathbf{x}}_{-i}$ denote the demand profile without that of user $i$. Given the offered price per unit of data ${p_i}$, the myopic user chooses the action that maximizes its utility. Then, the basic utility of the user is formulated as follows:
\begin{equation}\label{Eq:1}
	{v_i}({x_i},{\mathbf{x}}_{-i},{p_i}) = {f_i}(x_i) + \sum\limits_{j \in \mathcal{N}} {{g_{ij}}} {x_i}{x_j} - {p_i}(x_i).
\end{equation}
${f_i}(x)$ represents the private utility which is an internal effect or gain that user $i$ obtains from the social data demand. It is a linear-quadratic function ${f_i}(x_i) = {a_i}{x_i} - {b_i}{x_i}^2$, where ${a_i}$ and ${b_i}$ are the coefficients that capture the intrinsic value of the data demand to different users with heterogeneity. $g_{ij}$ indicates the external benefits due to the network effects. In social networks, one user can enjoy an additional benefit from the actions of other users~\cite{easley2010networks}. In particular, $g_{ij}$ refers to the influence of user $j$ on user $i$, which we assume to be unidirectional. In other words, $g_{ij} = g_{ji}$ represents the social tie between users $i$ and $j$, i.e., the social tie is reciprocal. Nevertheless, the same model can be applied to bidirectional social relations straightforwardly. Moreover, ${g_{ii}} = 0$ which means one user cannot influence oneself. The fee that the MNO charges to user $i$ is equal to ${p_i}{x_i}$, i.e., usage-based pricing. In this paper, we consider the case that the MNO can charge the different users with different prices, i.e., the discriminatory pricing scheme~\cite{kang2012price, tushar2017price}.

Moreover, the users may experience congestion with the increase of social data demand at the same time, e.g., service delay, due to the limited radio resources in mobile networks. Therefore, we investigate the users' behaviors by jointly incorporating the network effects and congestion effects. In particular, we use another term $\frac{c}{2}{\left(\sum\nolimits_{j \in {\cal N}} {{x_j}} \right)^2}$ to denote the congestion effects, where $\frac{c}{2}$ is the coefficient of congestion. The quadratic sum form $\frac{c}{2}{\left({\sum _{j \in {\cal N}}}{x_j}\right)^2}$ indicates that all users experience and are affected by the congestion, and the marginal cost of the congestion increases as the total demand increases~\cite{gong2015network, zhang2017motivating}. Then, the net utility of user $i$ extended from (\ref{Eq:1}) is expressed as follows:
\begin{equation}\label{Eq:2}
	{u_i}({x_i},\mathbf{{x}}_{-i},{p_i}) = {a_i}{x_i} - {b_i}{x_i}^2 + \sum\limits_{j \in \mathcal{N}} {{g_{ij}}} {x_i}{x_j} - \frac{c}{2}{\left(\sum\limits_{j \in \mathcal{N}} {{x_j}} \right)^2} - {p_i}{x_i}.
\end{equation}
The objective of the MNO is to maximize its revenue which is expressed as follows:
\begin{equation}\label{Eq:3}
	\max_{ p_i } \sum\limits_{i \in \cal N} {{p_i}{x_i}}.
\end{equation}

Naturally, the two-stage Stackelberg game can be adopted to model the interaction between the MNO and users~\cite{zhang2016social, zhang2017motivating}. In the upper Stage I, the MNO, i.e., the leader, determines price $p_i$ to maximize its revenue. In the lower Stage II, the users, i.e., the followers, decides on their individual data demand $x_i$ in order to maximize their utilities being aware of price $p_i$ set by the MNO. Using the backward induction methods, the existence and uniqueness of a set of strategies where no user deviates based on the given price, i.e., the Nash equilibrium, is investigated first. Based on this Nash equilibrium, the optimal pricing of the MNO can be further addressed.

\subsection{Dynamic model extension}
Here we propose a discrete-time sequential dynamic pricing scheme for the revenue maximization of the MNO over multiple time periods, e.g., days or weeks. We first define ${\bf{p}}^{(k)} = \left[ p_1^{(k)} , \ldots , p_i^{(k)},  \ldots, p_N^{(k)} \right]^\top $ as the instantaneous pricing strategies, where $p_i^{(k)}$ is the price charged to the user $i$ in time period $k$. Let $x_i^{(k)}$ be the social data demand from user $i$ given the price $p_i^{(k)}$. The price vector can be varied over time periods $t$. In each time period, the MNO approaches the users in a certain order for charging which is denoted by ${{\cal R}} \buildrel \Delta \over = \{ r_1, \ldots ,r_N \}$, where $r_i \in {\cal N}$ and $r_i \ne r_j$, $\forall i \ne j$.

We assume that the users have complete information about the past history of the demand of all users in the network. This information can be obtained through a long-term learning or side channel, e.g., statistics from the operator itself or from third party marketing firm~\cite{sinkula1994market}. The information known by user $i$ in the time period $k$ can be represented by tuples, i.e., $\mathcal{I}_{i}^{(k)}= \left(x_i^{(1:k-1)},p_i^{(1:k-1)},{\bf{x}}_{ - i}^{(1:k - 1)},{[x_j^{(k)}]_{j \in {\cal N}}}\right)$, where $x_i^{(1:k-1)}$ is the set of historical demand of user $i$ from time period $1$ to period $k-1$, i.e., $x_i^{(1:k-1)} = \{ x_i^{k'} \}_{k'=1,\ldots,k-1}$. Similarly, $p_i^{(1:k-1)}$ is a set of historical price offered to user $i$, i.e., $p_i^{(1:k-1)} = \{ p_i^{k'} \}_{k'=1,\ldots,k-1}$, and ${\bf{x}}_{ - i}^{(1:k - 1)}$ is a set of historical demands of all other users except user $i$, i.e., ${\bf{x}}_{ - i}^{(1:k - 1)} = \{ {\bf{x}}_{ - i}^{k'} \}_{k'=1,\ldots,k-1}$. Let $y_i^{(k)} = \sum\limits_{k' = 1}^k {{x_i^{(k')}}}$ be the cumulative demand (i.e., quantity of social data) purchased by user $i$ from time period $1$ to period $k$. In the time period $k$, user $i$ aims to maximize its expected utility, which is formulated as
\begin{multline}\label{Eq:4}
{u_i}\left(x_i^{(k)},{\mathcal{I}}_{i}^{(k)},p_i^{(k)}\right) = {a_i}\left(y_i^{(k - 1)} + x_i^{(k)}\right)  - {b_i}{\left(y_i^{(k - 1)} + x_i^{(k)}\right)^2} + {\sum _{j \in \cal N}}{g_{ij}} \left(y_j^{(k-1)}+x_j^{(k)}\right)\\ \times \left(y_i^{(k - 1)} + x_i^{(k)}\right)- \frac{c}{2} {\sum\limits_{j \in \cal N} \left(y_j^{(k - 1)} + x_j^{(k)} \right)^2}   - \sum\limits_{t = 1}^k {p_i^{(t)}x_i^{(t)}}.
\end{multline}

The objective of the MNO is to maximize its cumulative revenue. In each time period $k$, the MNO chooses $p_i^{(k)}$ to maximize its revenue in the current time period. This revenue is given by $\sum\limits_{i \in {\cal N}} {p_i^{(k)}x_i^{(k)}}$. Note that this revenue maximization is the dynamic optimization. This is due to the fact that the revenue of the MNO depends on the demand of all users. The demand of each user depends on the historical demand profile of the other users, i.e., $\mathcal{I}_{i}^{(k)}$, as expressed in (\ref{Eq:4}). For example, if one user knows that its social neighbors have accessed a certain video, that user is likely to access the same video content.

The sequential dynamic pricing scheme~\cite{zhou2017optimal} can be obtained from solving the cumulative revenue maximization problem for the MNO, which is given by
\begin{equation}\label{Eq:5}
	\max {\Pi _{\cal R}} =\max \sum\limits_{k = 1}^\infty  {\sum\limits_{i \in \cal N} {p_{{r_i}}^{(k)} x_{{r_i}}^{(k)}} },
\end{equation}
and the pricing is run for a long period of time which is equivalent to time infinity. Recall that in each time period $k$, the revenue maximization is expressed by $\max \sum\limits_{i \in \cal N} {p_{{r_i}}^{(k)}x_{{r_i}}^{(k)}}$, which can be divided into sub-problems of maximizing the revenue gained from each user $i$. The first sub-problem is to determine an optimal price charging to each user given a fixed order. Then, the second sub-problem is to determine an optimal order of users to be visited by the MNO.
\begin{Sup}
\begin{equation}\label{Eq:6}
	\Pi _{\cal R}^{(k)} = \max_{ p_{{r_i}}^{(k)} } 	\sum\limits_{i \in {\cal N}} p_{{r_i}}^{(k)} \widetilde{x}_{{r_i}}^{(k)},
\end{equation}
\qquad \qquad \qquad \qquad \qquad \emph{\textbf{subject to}} $\widetilde{x}_{{r_i}}^{(k)} = \arg \max_{ \acute{x}_{{r_i}}^{(k)} }  {u_{{r_i}}}( \acute{x}_{{r_i}}^{(k)},{\cal I}_{{r_i}}^{(k)},p_{{r_i}}^{(k)}), \forall i,$ \\
\emph{where the order of users $\mathcal{R} = \{ r_1, \ldots ,r_i, \ldots , r_N\}$ is fixed.}
\end{Sup}

\begin{Sup}
\begin{equation}\label{Eq:7}
\max_{ {\cal R} } \Pi_{\cal R}^{(k)},
\end{equation}
 \qquad \qquad \qquad \qquad  \qquad \qquad \quad \quad\emph{\textbf{subject to}} ${\cal R} \in {\cal Q} ({\cal N})$,
\\\emph{where ${\cal Q} ({\cal N})$ is the set of all possible orders of the users in ${\cal N}$. In other words, ${\cal Q} ({\cal N})$ is the set of all possible orders of users.}
\end{Sup}
The Sub-problem~$2$ can be regarded as a dynamic optimization problem. In particular, the MNO has to select the order of users to be visited and the prices charged to each user to maximize its revenue. After fixing the order $\cal R$, we denote ${\widetilde {{x_i}}^{(k)}}$ and ${\widetilde {{p_i}}^{(k)}}$ as the optimal demand of user $i$ and optimal price charged by MNO to user $i$ in time period $k$. However, the number of possible orders makes the problem complicated and intractable. Specifically, in each time period, we need to check $n!$ orders to determine an optimal one. Nevertheless, we observe some special structure characterized by the theorems given in the next section.

\section{Sequential Dynamic Pricing Scheme}\label{Sec:Seqpricing}

\subsection{Problem formulation}
In what follows, we denote ${\bf{\Lambda}}  = diag\left(\left[2{b_1},2{b_2}, \ldots ,2{b_N}\right]^\top\right)$, ${{\bf{\Lambda}} _c} = diag\left(\left[\frac{c}{2},\frac{c}{2}, \ldots ,\frac{c}{2}\right]^\top\right)+ {\bf{\Lambda}}  = diag\left(\left[2{b_1} + \frac{c}{2},2{b_2} + \frac{c}{2}, \ldots ,2{b_N} + \frac{c}{2}\right]^\top\right)$. We further denote ${\bf{a}} = [a_1, a_2, \ldots, a_N]^\top >{\bf{0}}$, and ${\bf{1}} = [1, 1, \ldots, 1]^\top$. In addition, ${\mathcal{C}}$ is an $N\times N $ matrix with $c$ being every element inside. ${\bf{I}}$ is an $N\times N $ identity matrix, and $G={\{g_{ij}}\}_{i, j \in \cal N}$.
\begin{assumption}
For all $i \in \cal N$, $2{b_i} > \sum\nolimits_{j \in \cal N} {\left( {{g_{ij}} - c} \right)} $.
\end{assumption}

Similar to that in~\cite{gong2015network, zhang2017motivating, xiong2017sponsor}, we make the above assumption to ensure the boundedness of the social data demand from each user, and we have the following theorem.

\begin{theorem}
The optimal revenue of the MNO in the time period $k$, $\Pi_{\cal R}^{(k)}$, and the optimal solution $ \widetilde{\bf{x}}_{\cal{R}}^{(k)}$ are unique and irrelevant to the order $\cal R$, provided that Assumption 1 is satisfied.
\end{theorem}

\begin{proof}
Firstly, we fix the order ${\cal {R}} = \{1, 2, \ldots, N\}$. Then, in the first time period the revenue is given by
\begin{equation}
{\Pi_{\cal{R}}} = \sum\limits_{i \in \cal N} {{p_i}{x_i}}  = \sum\limits_{i \in \cal N} {{x_i}\left({a_i} - 2{b_i}{x_i} + {\sum \limits_{j \in \cal N}}{g_{ij}}{x_j} - c\sum\limits_{j \in \cal N} {{x_j}} \right)}.
\end{equation}
Let $\widetilde x_i$ denote the optimal social data demand from user $i$. Here for ease the presentation, we omit $^{(1)}$ for the derivation of the first period. We take $\frac{{\partial {\Pi_{\cal{R} }  }}}{{\partial {\widetilde x_i}}} = 0$ according to the first-order condition. Due to $g_{ij}=g_{ji}$, we obtain
\begin{equation}
{a_i} - 4{b_i}{\widetilde x_i} + {\sum _{j \in \cal N}}{g_{ij}}{{\widetilde {x}}_j} - c\sum\limits_{j \in \cal N} {{\widetilde x_j}}  - c{\widetilde x_i} = 0,
\end{equation}
which is irrelevant to the order $\cal R$. Then we obtain the expression
\begin{equation}
\widetilde {\bf{x}} = {\left(2{{\bf{\Lambda }}_{\bf{c}}} - {\bf{G}} + {\mathcal{C}}\right)^{ - 1}}{\bf{a}}.
\end{equation}
Similarly, the optimal revenue of the MNO for the first time period is irrelevant to the order $\cal R$, which is given by $\sum\limits_{i \in \cal N} {{x_i}\left({a_i} - 2{b_i}{x_i} + {\sum _{j < i}}{g_{ij}}{x_j} - c\sum\limits_{j < i} {{x_j}} \right)}$. Then, we suppose the obtained optimal revenue holds in subsequent time periods $k' = 2, \ldots, k-1$, and we denote $x_i^{(k')}$ as the optimal demand of user $i$ in the time period $k'$, and $y_i^{(k)}$ be $\sum\limits_{k' = 1}^k {{x_i^{(k')}}}$. Thus, the revenue of the MNO in the time period $k$ is given as follows:
\begin{eqnarray}\label{Eq:8}
{\Pi ^{(k)}} &=& \sum\limits_{i \in \cal N} {p_i^{(k)}x_i^{(k)}}
=\sum\limits_{i \in \cal N} {x_i^{(k)}\left({a_i} - 2{b_i}y_i^{(k)} + {\sum\limits_{j \in \cal N}}{g_{ij}}y_j^{(k)} - c\sum\limits_{j \in \cal N} {y_j^{(k)} } \right)} \nonumber \\
&=&\sum\limits_{i \in \cal N} x_i^{(k)} \Bigg({a_i} - 2{b_i}y_i^{(k - 1)} - 2{b_i}x_i^{(k)}  + {\sum \limits_{j \in \cal N}}{g_{ij}}y_j^{(k - 1) } + {\sum \limits_{j \in \cal N}}{g_{ij}}x_j^{(k)} \nonumber \\&&- c\sum\limits_{j \in \cal N} {y_j^{(k - 1)}  }  - c\sum\limits_{j \in \cal N} {x_j^{(k)} } \Bigg).
 \end{eqnarray}

Then, we let $a_i^{(k)}$ be ${a_i} - 2{b_i}y_i^{(k - 1)} + {\sum _{j \in \cal N}}{g_{ij}}y_j^{(k - 1)} - c\sum\limits_{j \in \cal N} {y_j^{{{(k - 1}})}}$. With Proposition 1, we obtain that ${{\bf{a}}^{(k)}}$ is expressed by
\begin{equation}
{{\bf{a}}^{(k)}} = {\big[{\bf{I}} - ({\bf{\Lambda }} - {\bf{G}} +  {\mathcal{C}}){(2{{\bf{\Lambda }}_c} - {\bf{G}} +  {\mathcal{C}})^{ - 1}}\big]^{k - 1}}{\bf{a}},
\end{equation}
which is a vector of $a_i^{(k)}$. According to the first-order condition, the optimal social data demand ${\widetilde {x_i}^{(k)}}$ in time period $k$ must satisfy that the derivative ${\partial {\Pi ^{(k)}}}/{{\partial \widetilde {x_i}}^{(k)}} = 0$. Thus, we have
\begin{equation}
{\widetilde{\bf{x}}^{(k)}}  = {(2{{\bf{\Lambda }}_c} - {\bf{G}} +  {\mathcal{C}})^{ - 1}}{{\bf{a}}^{(k)}},
\end{equation}
which is also irrelevant to the order $\cal R$. Moreover, as the number of time periods goes to infinity, it is observed that the total optimal demand of all users is shown as follows:
\begin{eqnarray}
\mathop {\lim }\limits_{k \to \infty } \widetilde{{ y}}^{(k)} &=& \mathop {\lim }\limits_{k \to \infty } \sum\limits_{k' = 1}^k {\widetilde {x_i} }^{(k')} \nonumber\\ &=& \sum\limits_{k = 1}^\infty  {{{(2{{\bf{\Lambda }}_c} - {\bf{G}} +  {\mathcal{C}})}^{ - 1}}} {{\bf{a}}^{(k)}}\nonumber\\
 &=& {(2{{\bf{\Lambda }}_c} - {\bf{G}} +  {\mathcal{C}})^{ - 1}}{\Big( {\bf{I}} - \left[{\bf{I}} - ({\bf{\Lambda }} - {\bf{G}} +  {\mathcal{C}}){(2{{\bf{\Lambda }}_c} - {\bf{G}} +  {\mathcal{C}})^{ - 1}}\right]\Big)^{ - 1}}{\bf{a}} \nonumber\\
 &=& {({\bf{\Lambda }} - {\bf{G}} +  {\mathcal{C}})^{ - 1}}{\bf{a}}.
\end{eqnarray}

Note that the expression above is based on Proposition 2.  The proof is now completed.
\end{proof}

\begin{proposition}
The matrix $(2{{\bf{\Lambda }}_c} - {\bf{G}} + {\mathcal{C}})$ has non-negative entries and is invertible, provided that Assumption 1 is satisfied.
\end{proposition}
\begin{proof}
Please refer to the Appendix-A for details.
\end{proof}

\begin{proposition}
The spectrum radius of $\big[(2{{\bf{\Lambda }}_{\bf{c}}}-{\bf{\Lambda }})(2{{\bf{\Lambda }}_{\bf{c}}} - {\bf{G}} + {\mathcal{C}})^{-1}\big]^2$ is smaller than 1, provided that Assumption 1 is satisfied.
\end{proposition}
\begin{proof}
Please refer to the Appendix-B for details.
\end{proof}

Based on Theorem 1, the complexity of solving Sub-problem 2 can be reduced in a large extent. Meanwhile, Sub-problem~1 can be solved with the optimal price which is irrelevant to the order obtained from Sub-problem 2. Additionally, following Theorem~1, we can immediately obtain Theorem~2 as below.

\begin{theorem}
Provided that Assumption 1 is satisfied, the optimal solution to the optimization issue given in~(\ref{Eq:5}) is unique which can be expressed as:
\begin{equation}
\widetilde {{{\bf{x}}}}^{(k)} = {\big[2{{\bf{\Lambda }}_{\bf{c}}} - {\bf{G}} +  {\mathcal{C}}\big]^{ - 1}}{\big[{\bf{I}} - ({\bf{\Lambda }} - {\bf{G}} +  {\mathcal{C}}){(2{{\bf{\Lambda }}_c} - {\bf{G}} + {\mathcal{C}})^{ - 1}}\big]^{k - 1}}{\bf{a}}.
\end{equation}
As the time period $k$ goes to $\infty$, the total social data demand $\widetilde {{y}}^{(k)} = \mathop {\lim }\limits_{k \to \infty } \sum\limits_{t = 1}^k {\widetilde {{x}}^{(t)}}$ converges to $({\bf{\Lambda }} - {\bf{G}} +  {\mathcal{C}})^{ - 1}{\bf{a}}$.
\end{theorem}

Following the above theoretical analysis, we propose the sequential pricing scheme, as shown in Algorithm~\ref{Algo:Pricing}. Similar to~(\ref{Eq:7}), the revenue under the proposed algorithm is expressed as follows:
\begin{equation}\label{Eq:9}
\Pi_d = \sum\limits_{k = 1}^\infty  {\sum\limits_{i \in {\cal N}} {\widetilde p_i^{(k)}\widetilde x_i^{(k)}} }.
\end{equation}
Additionally, we define the total utilities of the users gained from social data demand as $\mathscr{U}_d$. The expression for $\mathscr{U}_d$ under the proposed algorithm is given as follows:
\begin{equation}\label{Eq:10}
\mathscr{U}_d = \sum\limits_{i \in \cal N} {\left({a_i}\widetilde {{y_i}} - {b_i}{{\widetilde {{y_i}}}^2} + \widetilde {{y_i}}{\sum _{j \in \cal N}}{g_{ij}}\widetilde {{y_j}} - \frac{c}{2}{{\sum\limits_{j \in \cal N} {\widetilde {{y_j}}} }^2}\right)}  - {\Pi _d}.
\end{equation}

To demonstrate the outperformance of our proposed pricing scheme, we also provide and compare the unique equilibrium solution under the static pricing strategy which is similar to that in~\cite{zhang2016social}. ${\Pi _s} = \sum\limits_{i \in \cal N} {\widehat {{p}}_i \widehat {{x}}_i }$~in (\ref{Eq:3}) and ${\mathscr U_s} = \sum\limits_{i \in \cal N} {{u_i}(\widehat{x}_i,{{\bf{\widehat x}}_{ - {\bf{ i}}}},{{\widehat p}_i})}$~in (\ref{Eq:2}) denote the MNO's revenue and the users' total utilities, respectively. After the comparison, we have the following theorem.

\begin{theorem}
The results obtained from the sequential dynamic pricing dominates the optimal static policy in terms of the revenue of the MNO and the total utilities of users provided that Assumption 1 is satisfied, i.e., ${\Pi _d} \ge {\Pi _s}$ and ${\mathscr U_d} \ge {\mathscr U_s}$.
\end{theorem}
\begin{proof}
Please refer to the Appendix-C for details.
\end{proof}

From Theorem 3, we have demonstrated that our alternative pricing scheme outperforms the static pricing proposed in~\cite{zhang2016social} in terms of the revenue of the MNO and the total utilities of the users. We further demonstrate the outperformance of the proposed dynamic pricing scheme in the simulation parts.
\begin{algorithm}[t]\footnotesize
 \caption{: Sequential dynamic pricing scheme for mobile social data market}
 \begin{algorithmic}[1]
 \REQUIRE ${\widetilde{x}^{(k)}} = [{\widetilde{x}_1^{(k)}},{\widetilde{x}_2^{(k)}}, \ldots {\widetilde{x}_N^{(k)}}]^\top.$
  \FOR{each time period $k$}
  \STATE \textbf{1)}: Obtain social data demand in the time period~$k$
  \begin{eqnarray*}
{\widetilde{\bf{x}}^{(k)}} = {[2{{\bf{\Lambda _c}}} - {\bf{G}} +  {\mathcal{C}}]^{ - 1}}{\bf{T}}^{k - 1}{\bf{a}},
\end{eqnarray*}
       where $ {\bf{T}} = {{\bf{I}} - ({\bf{\Lambda }} - {\bf{G}} +  {\mathcal{C}}){(2{{\bf{\Lambda }}_c} - {\bf{G}} +  {\mathcal{C}})^{ - 1}}}.$
  \STATE \textbf{2)}: Obtain total social data demand until time period~$k$:
  \begin{eqnarray*}
{\widetilde{y}_i^{(k)}} = \sum\limits_{k' = 1}^{k - 1} {\widetilde{x}_i^{(k')}}.
   \end{eqnarray*}
  \STATE \textbf{3)}: Set the order of users as $\mathcal{R} = {1,2, \ldots, N}$.
  \STATE \textbf{4)}: Obtain the prices charged to users for time period~$k$:
\begin{eqnarray*}
{\widetilde{p}_i^{(k)}} = {a_i} - 2{b_i}\left({\widetilde{y}_i^{(k-1)}} +{\widetilde{x}_i^{(k)}}\right) + \sum\limits_{j \in \cal N} {{g_{ij}}\widetilde y_j^{(k - 1)}}  + \sum\limits_{j < i} {{g_{ji}}\widetilde x_j^{(k)}}  - c\widetilde y_i^{(k - 1)} - c\sum\limits_{j < i} {\widetilde{x}_j^{(k)}}.
\end{eqnarray*}
  \STATE \textbf{Output}: $[{\widetilde{p}_1^{(k)}},{\widetilde{p}_2^{(k)}}, \ldots, {\widetilde{p}_N^{(k)}}]^\top$
  \ENDFOR
 \end{algorithmic} \label{Algo:Pricing}
\end{algorithm}

\subsection{Reordering for social fairness}

Recall from Theorem~1, we have verified that the proposed Algorithm~1 is optimal towards the revenue maximization of the MNO and irrelevant to the order of users to be visited by the MNO. However, it is worth noting that each user's individual utility may differ considerably. Therefore, the social fairness among mobile users in the market under our consideration needs to be ensured. The social fairness is an important measure from user's perspective in a market pricing problem~\cite{laffont1998network}. Generally, the fairness indicates that the utility of all users is distributed in a fair manner over time. In turn, the MNO improves its popularity and reputation degree among consumers by ensuring the social fairness. Since the optimal revenue of the MNO keeps unchanged irrespective of the order of users to be visited. In the following, we leverage the property of this ``invariance of ordering'' to ensure the fairness amongst mobile users. Specifically, we utilize the \textit{max-min} fairness allocation concept to refine the sequential dynamic pricing scheme. The goal is to find the optimal order of users to be visited by the MNO in order to ensure the max-min fairness among users with respect to the individual utilities of users.

We first consider the specific price charging to different users in our proposed sequential dynamic pricing scheme, i.e., Algorithm 1. We assume the MNO visits its users in the same fixed order, which is denoted by $\{1, 2,\ldots, N={\left| {\cal N} \right|}\}$. Recall that the optimal social data demand of users in time period $k$ is given as
\begin{eqnarray}
{\widetilde {\bf{x}}^{(k)}} &=& (2{{\bf{\Lambda }}_c} - {\bf{G}} + {\cal C})^{ - 1}\left[{\bf{I}} - ({\bf{\Lambda }} - {\bf{G}} + {\cal C}){(2{{\bf{\Lambda }}_c} - {\bf{G}} + {\cal C})^{ - 1}}\right]^{k - 1}{\bf{a}} \nonumber \\ &=& {\left( {\frac{{2{\Lambda _c} - \Lambda }}{{2{{\bf{\Lambda }}_c} - {\bf{G}} + {\cal C}}}} \right)^{k - 1}}{(2{{\bf{\Lambda }}_c} - {\bf{G}} + {\cal C})^{ - 1}}\bf{a}.
\end{eqnarray}
The optimal pricing of the MNO in the time period $k$ is given as
\begin{eqnarray}
{\widetilde p^{(k)}} &=& {[{\bf{I}} - ({\bf{\Lambda }} - {\bf{G}} + {\cal C}){(2{{\bf{\Lambda }}_c} - {\bf{G}} + {\cal C})^{ - 1}}]^{k - 1}}{\bf{a}} - ({\bf{\Lambda }} - {\bf{G}} + {\cal C}){{\bf{x}}^{(k)}}\nonumber \\ &=& \left[ {{\bf{I}} - ({\bf{\Lambda }} - {\bf{G}} + {\cal C}){{(2{{\bf{\Lambda }}_c} - {\bf{G}} + {\cal C})}^{ - 1}}} \right]{\left[ {{\bf{I}} - ({\bf{\Lambda }} - {\bf{G}} + {\cal C}){{(2{{\bf{\Lambda }}_c} - {\bf{G}} + {\cal C})}^{ - 1}}} \right]^{k - 1}}{\bf{a}}\nonumber \\ &=& \left( {\frac{{2{\Lambda _c} - \Lambda }}{{2{{\bf{\Lambda }}_c} - {\bf{G}} + {\cal C}}}} \right){\left( {\frac{{2{\Lambda _c} - \Lambda }}{{2{{\bf{\Lambda }}_c} - {\bf{G}} + {\cal C}}}} \right)^{k - 1}}{\bf{a}}.
\end{eqnarray}

Note that in the case of a completely symmetric graph, where $a_i=a$, $b_i = b$, $g_{ij}=g$, $\forall i \in \cal N$. The optimal data demand of the user in time period $k$ in Algorithm 1, is given as:
\begin{equation}
{\widetilde x_i^{(k)}}  = {\left( {\frac{{2b + c}}{{4b - (N - 1)g + Nc}}} \right)^{k - 1}}\frac{a}{{4b - (N - 1)g + Nc}},
\end{equation}
where
\begin{equation}
{\frac{{2b + c}}{{4b - (N - 1)g + Nc}}} < 1.
\end{equation}

Thus, the price charging user $i$ in time period $k$, which is the $m$th user visited by the MNO, is given by
\begin{equation}\label{Eq:linearprice}
\widetilde p_i^{(k)} = {\left( {\frac{{2b + c}}{{4b - (N - 1)g + Nc}}} \right)^{k - 1}}\frac{{a\left[2b + c - (N - m)g\right]}}{{4b + Nc- (N - 1)g}}.
\end{equation}

It is clear from~(\ref{Eq:linearprice}) that the price charging to users increases with the increase of visiting orders. This leads to inter-individual disparity among users since prices vary widely. Therefore, the fair network utility allocation in the pricing shown in Algorithm~1 is not ensured. To alleviate this issue, we make a simple change in Algorithm~1 for selecting the visiting orders in each time period. Specifically, the MNO chooses any random order in the first time period, and selects the order based on the predefined max-min fairness criteria in subsequent time periods~\cite{lin2017performance}. In the case of completely symmetric graph, for the time period $k>1$, the order ${{\cal R}} \buildrel \Delta \over = \{ r_1, \ldots ,r_N \}$ is selected such that the following condition
\begin{equation}
r_i  = \arg \min_{i \in {\cal N} \backslash \{r_1, r_2, \ldots, r_{i-1}\}}{u_i}\left( {\left[ {\widetilde x_i^{(k')}} \right]_{k' = 1}^{k - 1},\left[ {\widetilde {\cal I}_i^{(k')}} \right]_{k' = 1}^{k - 1},\left[ {\widetilde p_i^{(k')}} \right]_{k' = 1}^{k - 1}} \right)
\end{equation}
holds, where ${u_i}\left(x_i^{(k)},{\mathcal{I}}_{i}^{(k)},p_i^{(k)}\right)$ is given in~(\ref{Eq:4}). In particular, the individual utility of the user in any time period $k$ which is the $m$th visited user, is given by
\begin{equation}
u_m = \frac{{{a^2}\left[ {2b + c - (N - m)g} \right]}}{{{{\left[ {4b + Nc- (N - 1)g} \right]}^2}}} + \frac{{{a^2}\left( {2b + c - mg} \right)}}{{{{\left[ {4b - (N - 1)g + Nc} \right]}^2}}}\frac{{{{\left( {\frac{{2b + c}}{{4b - (N - 1)g + Nc}}} \right)}^2}}}{{1 - {{\left( {\frac{{2b + c}}{{4b + Nc- (N - 1)g}}} \right)}^2}}}.
\end{equation}
This equation indicates that the user visited in the beginning has a higher individual utility compared with its counterparts which are visited in the end. This simple but effective scheme leads to the max-min fairness over time, with respect to allocating network utility of users.

\section{Simultaneous Dynamic Pricing Scheme}\label{Sec:Simultaneous}
In Section~\ref{Sec:Seqpricing}, we propose the sequential dynamic pricing scheme where the social data demand decision of users responds in different time scales. To obtain more insights from general dynamics where the users choose their strategies in the same time scale, we further explore a simultaneous pricing scheme.

In the simultaneous dynamic pricing scheme, the MNO determines the price at the beginning of each time period. Then, the mobile users decide on their individual social data demand simultaneously, taking the network effects in the social domain and the congestion effects in the network domain into account.

\subsection{Problem formulation}
We consider a finite selling time periods, $T$. Let ${\bf{p}}^{(k)} = \left[ p_1^{(k)} , \ldots , p_i^{(k)},  \ldots, p_N^{(k)} \right]^\top $ still be the pricing strategies determined by the MNO, where $p_i^{k}$ denotes the price charged to user $i$ in time period $k$. Given the price, the users interact with each other due to the network effects and congestion effects, as discussed in Section~\ref{Sec:Model}. In particular, we consider that the users do not interact with each other in the same time period because of the slow influence spread~\cite{chen2009efficient}. However, on one hand, the users are still influenced by the total social data demand of their social neighbours in previous time periods due to the network effects. On the other hand, the users are influenced by the past total social data demand of all the users due to the congestion effects.

Recall that $x_i^{(k)}$ denotes the social data demand of user $i$ in response to price $p_i^{(k)}$ and $y_i^k$ denotes the total social data demand of user $i$ until time period $k$. We still adopt the linear quadratic function to formulate the utility of the users. Each user $i \in \cal N $ in the time period $k$ decides on its social data demand to maximize its expected utility in this time period. Specifically, this is formulated as follows:
\begin{multline}\label{Eq:Simul}
{u_i}\left(x_i^{(k)},{\mathcal{I}}_{i}^{(k)},p_i^{(k)}\right) = {a_i}\left(y_i^{(k - 1)} + x_i^{(k)}\right)  - {b_i}{\left(y_i^{(k - 1)} + x_i^{(k)}\right)^2} + {\sum _{j \in \cal N}}{g_{ij}} y_j^{(k-1)} \left(y_i^{(k - 1)} + x_i^{(k)}\right)\\- \frac{c}{2} {\sum\limits_{j \in \cal N} \left(y_j^{(k - 1)} \right)^2}   - \sum\limits_{t = 1}^k {p_i^{(t)}x_i^{(t)}}.
\end{multline}
It is worth noting that the user's utility in time period $k$ depends only on other users' social data demand in previous $k-1$ time periods. Thus, (\ref{Eq:Simul}) is different from (\ref{Eq:4}) in the term that captures the influence within the current time period.

Therefore, the revenue maximization problem faced by the MNO can be written as follows:
\begin{equation}
\max_{\left\{{\bf{p}}^{(k)}\right\}} \sum\limits_{k \in \cal T}  {\sum\limits_{i \in \cal N} {p_{{i}}^{(k)} x_{{i}}^{(k)}} },
\end{equation}
where $\cal T$ denotes the selling time periods under consideration. We then analyze the optimal pricing and the corresponding social data demand in the following.

\subsection{Optimal pricing for slow influence spread}
According to the first-order condition, by setting the derivative $\frac{ {\partial{u_i}\left( {x_i^{(k)},{\cal I}_i^{(k)},p_i^{(k)}} \right)}}{{\partial x_i^{(k)}}} = 0,$ we obtain the optimal social data demand of user $i$ in time period $k$ for maximizing its individual utility in~(\ref{Eq:Simul}), which is given as $x_i^{(k)}= \max \left\{ {0,\frac{{{a_i} + \sum\limits_{j \in \cal N} {{g_{ij}}y_j^{(k - 1)}}  - c\sum\limits_{j \ne i} {y_j^{(k - 1)}}  - p_i^{(k)}}}{{2{b_i} + c}} - y_i^{(k)}} \right\}$.

Here we first impose the positivity constraint on the social data demand, and then validate that the optimal social data demand is indeed positive. Under this assumption, we have $(2{b_i} + c)y_i^{(k)} = {a_i} + \sum\limits_{j \in \cal N} {{g_{ij}}y_j^{(k - 1)}}  - c\sum\limits_{i \ne j} {y_j^{(k - 1)}}  - p_i^{(k)}$.
Thereafter, the optimal price charged to user $i$ in time period $k$ is given as $p_i^{(k)} = {a_i} + \sum\limits_{j \in \cal N} {{g_{ij}}y_j^{(k - 1)}}  - c\sum\limits_{j \in \cal N} {y_j^{(k - 1)}}  - 2{b_i}y_i^{(k)}$.

Therefore, the total revenue of the MNO can be written as
\begin{equation}
\Pi  = \sum\limits_{k' = 1}^T {\sum\limits_{i \in \cal N} {p_i^{(k')}x_i^{(k')}} }
=  \sum\limits_{k' = 1}^T {\sum\limits_{i \in \cal N} {\left( {{a_i} + \sum\limits_{j \in \cal N} {{g_{ij}}y_j^{(k' - 1)}}  - c\sum\limits_{j \in\cal N} {y_j^{(k' - 1)}}  - 2{b_i}y_i^{(k')}x_i^{(k')}} \right)x_i^{(k')}} }.
\end{equation}

In what follows, we continue to use the first-order optimality condition, i.e., $\frac{{\partial \Pi }}{{\partial x_i^{(k)}}} = 0$ for $t = 1, \ldots, T$ and $\forall i \in \cal N$. Then, we have $\frac{{\partial\left( \sum\limits_{k' = 1}^T {\sum\limits_{i \in \cal N} {\left( {{a_i} + \sum\limits_{j \in \cal N} {{g_{ij}}\sum\limits_{s = 1}^{k' - 1} {x_j^{(s)}} }  - c\sum\limits_{j \in \cal N} {\sum\limits_{s = 1}^{k' - 1} {x_j^{(s)}} }  - 2{b_i}\sum\limits_{s = 1}^{k'} {x_i^{(s)}} } \right)x_i^{(k')}} } \right) }}{{\partial x_i^{(k)}}} = 0$.
With simple transformations, we have the following equation for $k = 1, \ldots, T$ and $\forall i \in \cal N$,
\begin{multline}
{a_i} - 4{b_i} - 2{b_i}\sum\limits_{s = 1}^{k - 1} {x_i^{(s)}}  - c\sum\limits_{j \in \cal N} {\sum\limits_{s = 1}^{k - 1} {x_j^{(s)}} }  + \sum\limits_{j \in \cal N} {{g_{ij}}\sum\limits_{s = 1}^{k - 1} {x_j^{(s)}} }  - 2{b_i}\sum\limits_{s = k + 1}^T {x_i^{(s)}} \\ - c\sum\limits_{j \in \cal N} {\sum\limits_{s = k + 1}^T {x_j^{(s)}} }  + \sum\limits_{j \in \cal N} {{g_{ij}}\sum\limits_{s = k + 1}^{k - 1} {x_j^{(s)}} } = 0.
\end{multline}
For brevity, we write the matrix form of the above equation shown as follows:
\begin{equation}
{\bf a} - 2{\bf \Lambda} {{\bf x}^{(k)}} - \sum\limits_{s = 1}^{k - 1} ({\bf \Lambda}  - {\bf G} + {\mathcal C}){\bf x}^{(s)}  - \sum\limits_{s = k + 1}^T {({\bf \Lambda}  - {\bf G} + {\mathcal C}){{\bf x}^{(s)}}} = \bf 0.
\end{equation}
Due to the symmetry of this set of equations, we know that ${{\bf x}^{(k)}}$ is the same for $\forall k = 1, \ldots, T$. Thus, we have ${\bf{a}} - 2{\bf{\Lambda }}{{\bf{x}}^{(k)}} - (T - 1)({\bf{\Lambda }} - {\bf{G}} + {\cal C}){{\bf{x}}^{(k)}} = 0$.
Accordingly, we obtain the final expression for the optimal social data demand of users in time period $k$: ${{\bf{x}}^{*(k)}} = \big( {\left( {T + 1} \right){\bf{\Lambda }} - \left( {T - 1} \right)\left( {{\bf{G}} - {\cal C}} \right)} \big)^{-1} {\bf{a}}$.
This derived data demand is always positive because of Proposition 1, which thus justifies our positivity assumption of the social data demand. Recall that we have
\begin{equation}\label{Eq:Simuprice}
p_i^{*(k)} = {a_i} + \sum\limits_{j \in \cal N} {{g_{ij}}y_j^{(k - 1)}}  - c\sum\limits_{j \in \cal N} {y_j^{(k - 1)}}  - 2{b_i}y_i^{(k)}.
\end{equation}
Accordingly we rewrite~(\ref{Eq:Simuprice}) in a matrix form, which is shown as follows:
\begin{eqnarray}
{{\bf p}^{*(k)}} &=& {\bf{a}} + (k - 1){\bf{G}}{{\bf{x}}^{(k)}} - k{\cal C}{{\bf{x}}^{(k)}} - k{\bf{\Lambda }}{{\bf{x}}^{(k)}}\\
 &=& {\bf{a}} + \left( {(k - 1){\bf{G}} - k({\cal C} + {\bf{\Lambda }})} \right){\left( {(T + 1){\bf{\Lambda }} - \left( {T - 1} \right)\left( {{\bf{G}} - {\cal C}} \right)} \right)^{ - 1}}{\bf{a}}\nonumber \\
 &=& \left( {(T - k + 1){\bf{\Lambda }} + (T - k - 1){\cal C} - (T - k){\bf{G}}} \right){\left( {(T + 1){\bf{\Lambda }} - \left( {T - 1} \right)\left( {{\bf{G}} - {\cal C}} \right)} \right)^{ - 1}}{\bf{a}}.\nonumber
\end{eqnarray}

From the previous discussions in this section, we conclude with the following theorem.
\begin{theorem}
Provided that Assumption 1 is satisfied, the optimal prices in the time period $k$ under simultaneous dynamic pricing scheme is unique, which is given by
\begin{equation}\small
{{\bf p}^{*(k)}} = \left( {(T - k + 1){\bf{\Lambda }} + (T - k - 1){\cal C} - (T - k){\bf{G}}} \right){\left( {(T + 1){\bf{\Lambda }} - \left( {T - 1} \right)\left( {{\bf{G}} - {\cal C}} \right)} \right)^{ - 1}}{\bf{a}}.
\end{equation}
Likewise, the optimal social data demand of users in response to the given optimal price is also unique, which is written as follows:
\begin{equation}\small
{{\bf{x}}^{*(k)}} = \big( {\left( {T + 1} \right){\bf{\Lambda }} - \left( {T - 1} \right)\left( {{\bf{G}} - {\cal C}} \right)} \big)^{-1} {\bf{a}}.
\end{equation}
\end{theorem}

Note that we can reformulate the optimal prices as a function of time to analytically investigate the time-varying property of the price. For ease of presentation, we first define $\Xi (T) = {\left( {I - \frac{{T - 1}}{{T + 1}}\left( {{\bf{G}} - {\cal C}} \right){{\bf{\Lambda }}^{ - 1}}} \right)^{ - 1}}{\bf{a}}$.
Thus, we have ${{\bf{x}}^{(k)}} = \frac{{\Xi (T)}}{{(T + 1){\bf{\Lambda }}}}$,
and ${{x}}_i^{^{(k)}} = \frac{1}{{(T + 1)}}\frac{{{\Xi _i}(T)}}{{2{b_i}}}$.
To be specific, we reformulate the optimal price as
\begin{eqnarray}\label{Eq:Simupricechange}
p_i^{(k)} &=& {a_i} + \sum\limits_{j \in \cal N} {{g_{ij}}y_j^{(k - 1)}}  - c\sum\limits_{j \in \cal N} {y_j^{(k - 1)}}  - 2{b_i}y_i^{(k)}\nonumber \\
 &=& {a_i} + \sum\limits_{j \in \cal N} {{g_{ij}}\frac{{k - 1}}{{T + 1}}\frac{{{\Xi _j}(T)}}{{2{b_j}}}}  - c\sum\limits_{j \in \cal N} {\frac{{k - 1}}{{T + 1}}\frac{{{\Xi _j}(T)}}{{2{b_j}}}}  - 2{b_i}\frac{k}{{T + 1}}\frac{{{\Xi _i}(T)}}{{2{b_i}}} \nonumber \\
 &=& {a_i} - \frac{1}{{T + 1}}\sum\limits_{j \in \cal N} {{g_{ij}}\frac{{{\Xi _j}(T)}}{{2{b_j}}}}  + \frac{1}{{T + 1}}c\sum\limits_{j \in \cal N} {\frac{{{\Xi _j}(T)}}{{2{b_j}}}}  \nonumber \\ &&+ \frac{k}{{T + 1}}\left( {\sum\limits_{j \in \cal N} {{g_{ij}}\frac{{{\Xi _j}(T)}}{{2{b_j}}} - c\sum\limits_{j \in \cal N} {\frac{{{\Xi _j}(T)}}{{2{b_j}}} - {\Xi _i}(T)} } } \right).
\end{eqnarray}

It is observed from~(\ref{Eq:Simupricechange}) that the optimal prices under simultaneous dynamic pricing scheme change linearly with the slope given by
\begin{equation}
{\Phi _i}(T) = {\sum\nolimits_{j \in \cal N} {{g_{ij}}\frac{{{\Xi _j}(T)}}{{2{b_j}}} - c\sum\nolimits_{j \in \cal N} {\frac{{{\Xi _j}(T)}}{{2{b_j}}} - {\Xi _i}(T)} } }.
\end{equation}
The slope ${\Phi _i}(T) $ can be positive or negative, which indicates that the offered price may increase or decrease over $T$ time periods for different users. Furthermore, in the social data market where the users have homogeneous utility, the optimal prices decrease with time, which is characterized in the following proposition.
\begin{proposition}
If $a_i = a$, $b_i = b$ for all the users $i \in \cal N$, ${\Phi_i}(T)$ is negative, and thus the optimal prices under simultaneous dynamic pricing scheme decrease with time, i.e.,
\[\Phi (T) < 0.\]
\end{proposition}
\begin{proof}
Please refer to the Appendix-D for details.
\end{proof}

\section{Performance Evaluation}\label{Sec:Simulation}
In this section, we conduct the simulations to illustrate the impacts of different parameters on the proposed dynamic pricing schemes. In particular, we simulate the social graph $G$ using the Erd\H{o}s-R\'enyi (ER) graph to capture the social properties. In the ER graph, the social tie between any two users exists with the same probability $P_e$. Furthermore, we simulate the real social network based on the real data trace from Brightkite~\cite{cho2011friendship}. Brightkite is an online social networking service with explicit undirected social relationship based on mobile phone users. We randomly choose $N$ users from the real dataset and construct the social network, where $N = 10, 15, \ldots, 50$. For each given number of users, $N$, we obtain the average results with 500 runs. Figure~\ref{Fig:RealData} shows the total number of social ties and probability of social tie versus the number of users in the real dataset. The social tie between a pair of users follows the normal distribution ${\cal N}(\mu_g, 1)$, if it exists. Otherwise, the social tie is set as 0. We set the internal parameters of social users, $a_i$ and $b_i$, as the following normal distribution ${\cal N}(\mu_a, 1)$ and ${\cal N}(\mu_b, 1)$, respectively. The default parameter values are set as follows: $N=50$, $P_e = 0.8$, $\mu_a=1$, $\mu_b=20$, $\mu_g = 8$ and $c=10$. We implement the dynamic pricing with $50$ time periods.

\begin{figure}
\begin{minipage}[t]{0.5\linewidth}
\centering
\includegraphics[width=2.0in]{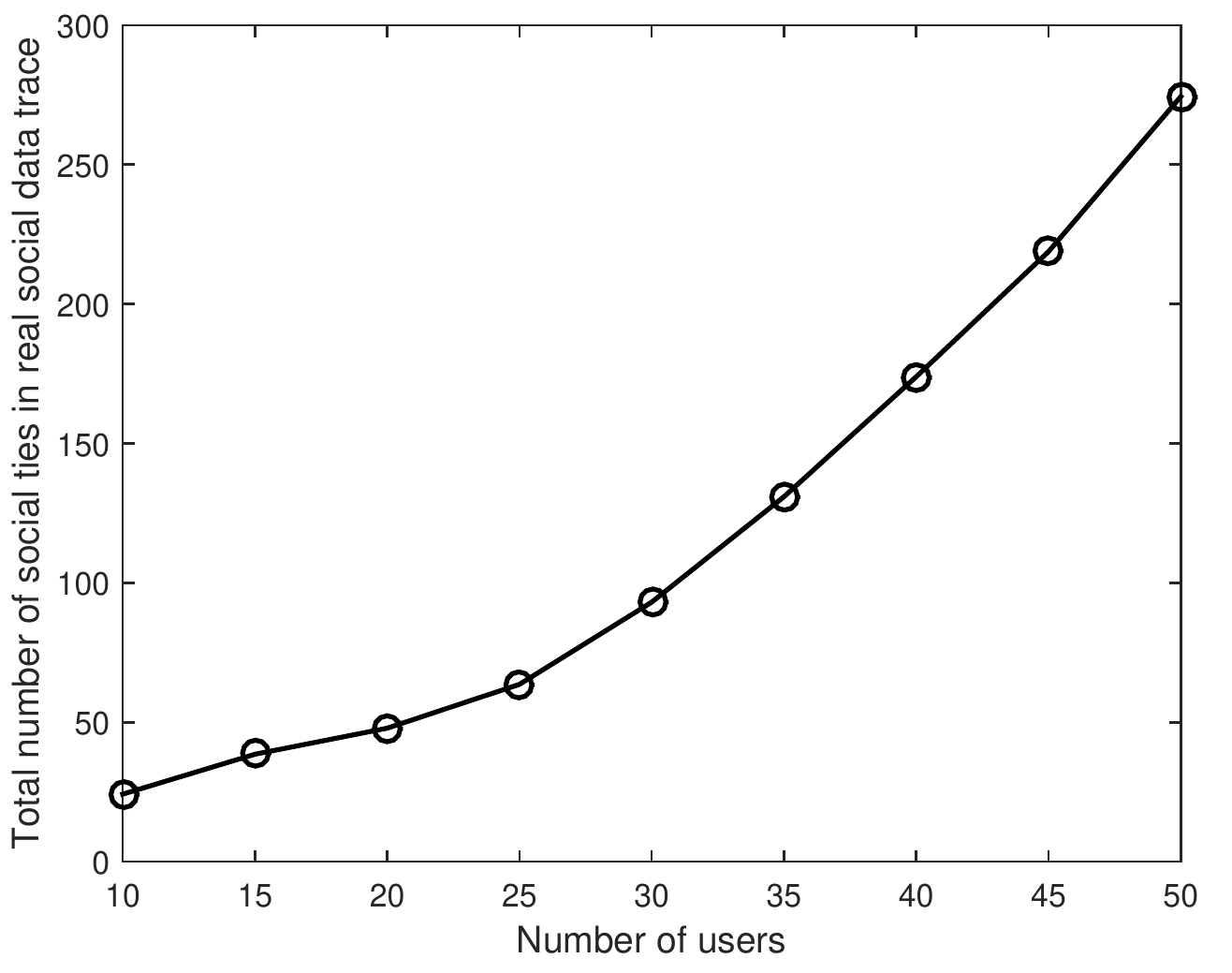}
\end{minipage}%
\begin{minipage}[t]{0.5\linewidth}
\centering
\includegraphics[width=2.0in]{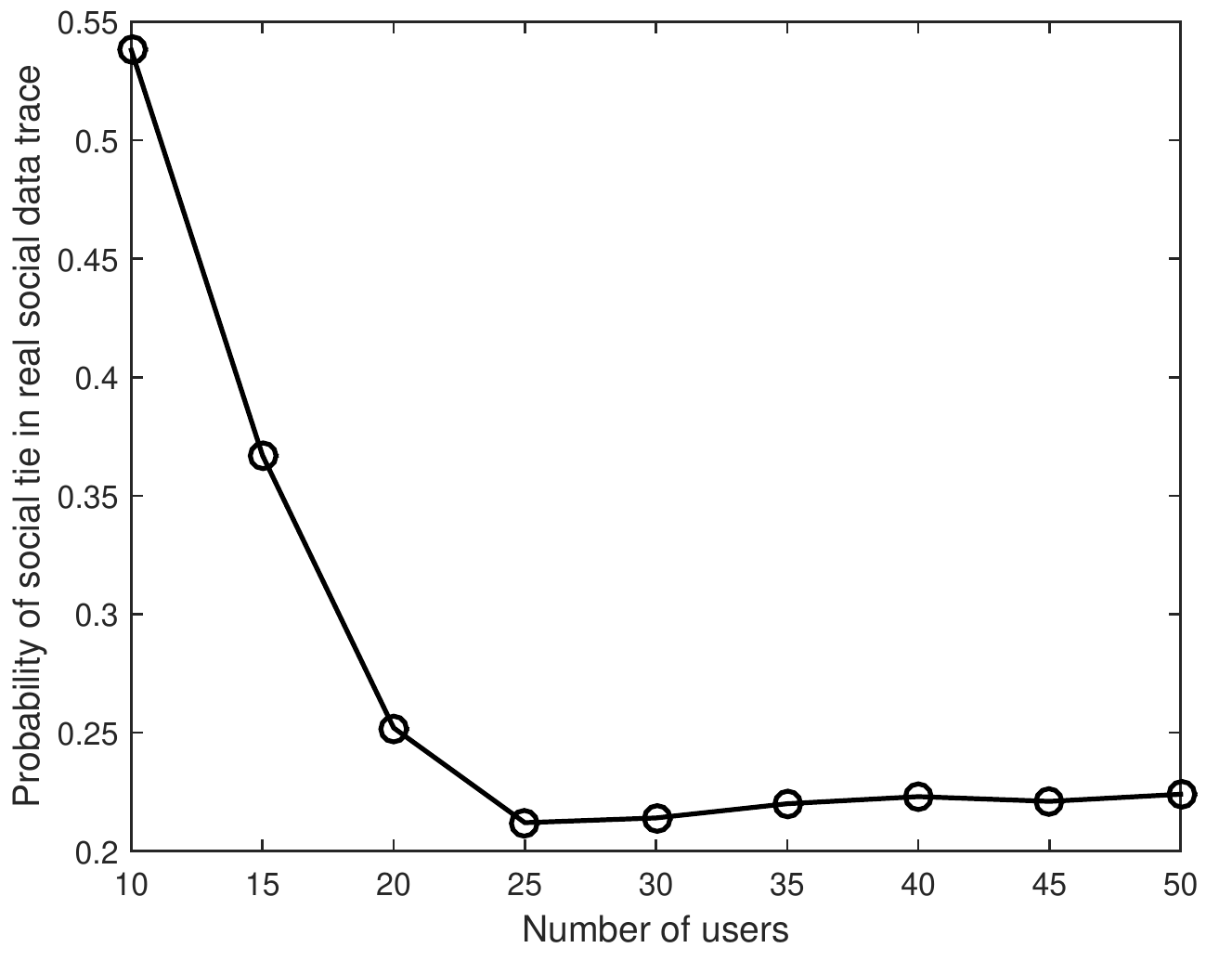}
\end{minipage}\caption{Real social data trace from Brightkite~\cite{cho2011friendship}: total number of social ties versus the number of users (left), and probability of social tie versus the number of users (right).}\label{Fig:RealData}
\vspace*{-4mm}
\end{figure}

\begin{figure}
\begin{minipage}[t]{0.325\textwidth}
\centering
\includegraphics[width=\textwidth]{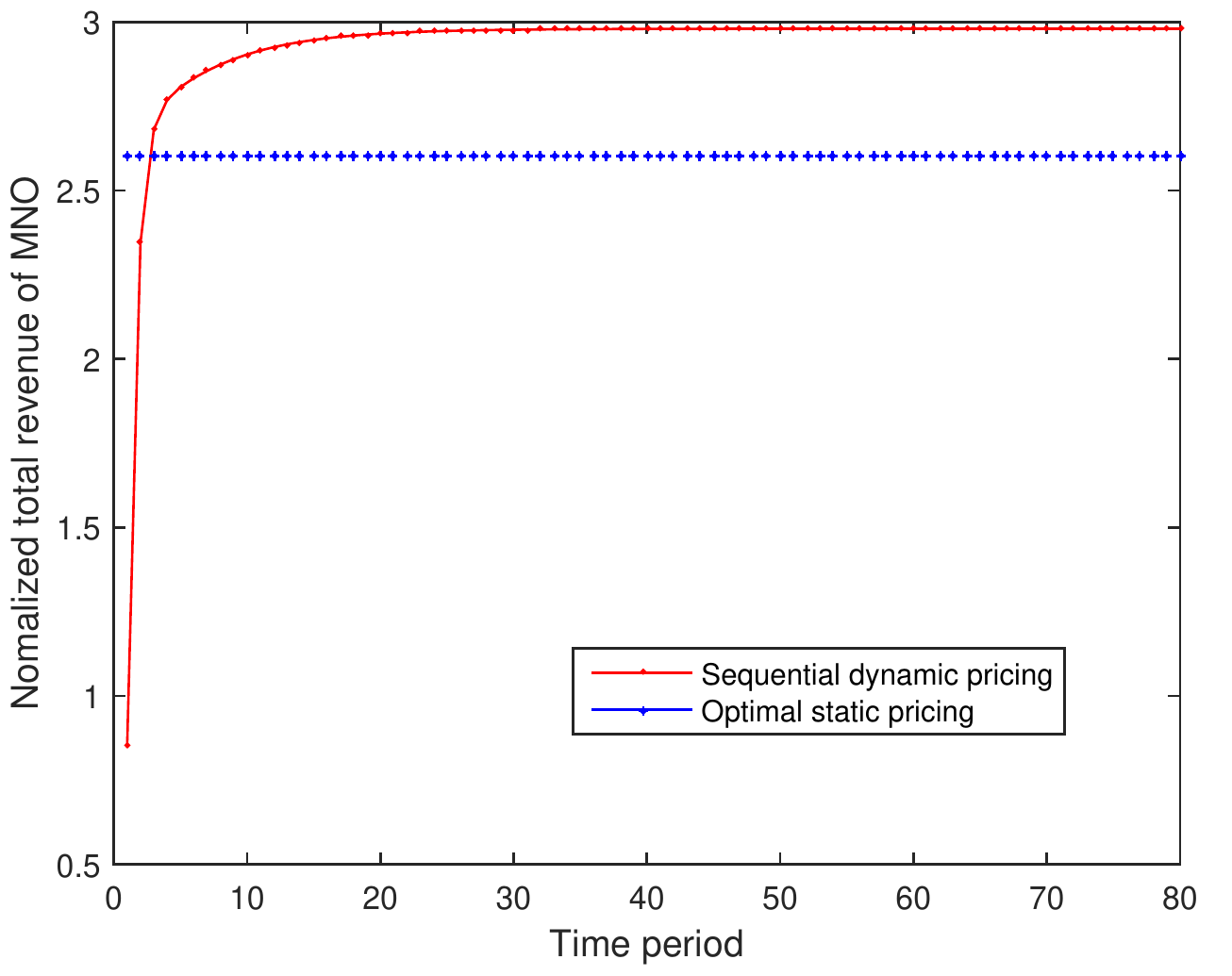}
\caption{Normalized total revenue of the MNO versus time periods.}
\label{fig:side:a}
\end{minipage}
\begin{minipage}[t]{0.33\textwidth}
\centering
\includegraphics[width=\textwidth]{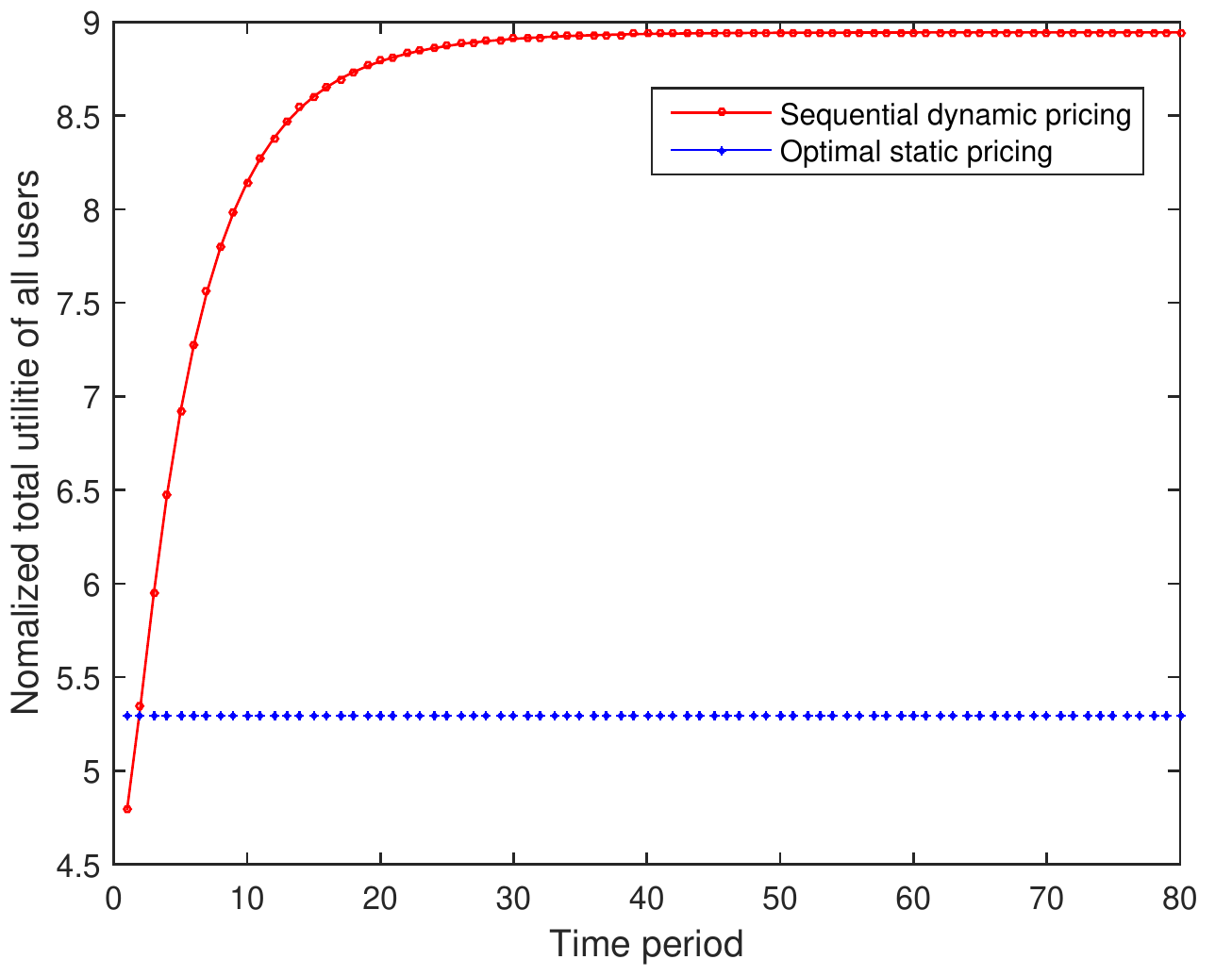}
\caption{Normalized total utilities of mobile users versus time periods.}
\label{fig:side:b}
\end{minipage}
\begin{minipage}[t]{0.33\textwidth}
\centering
\includegraphics[width=\textwidth]{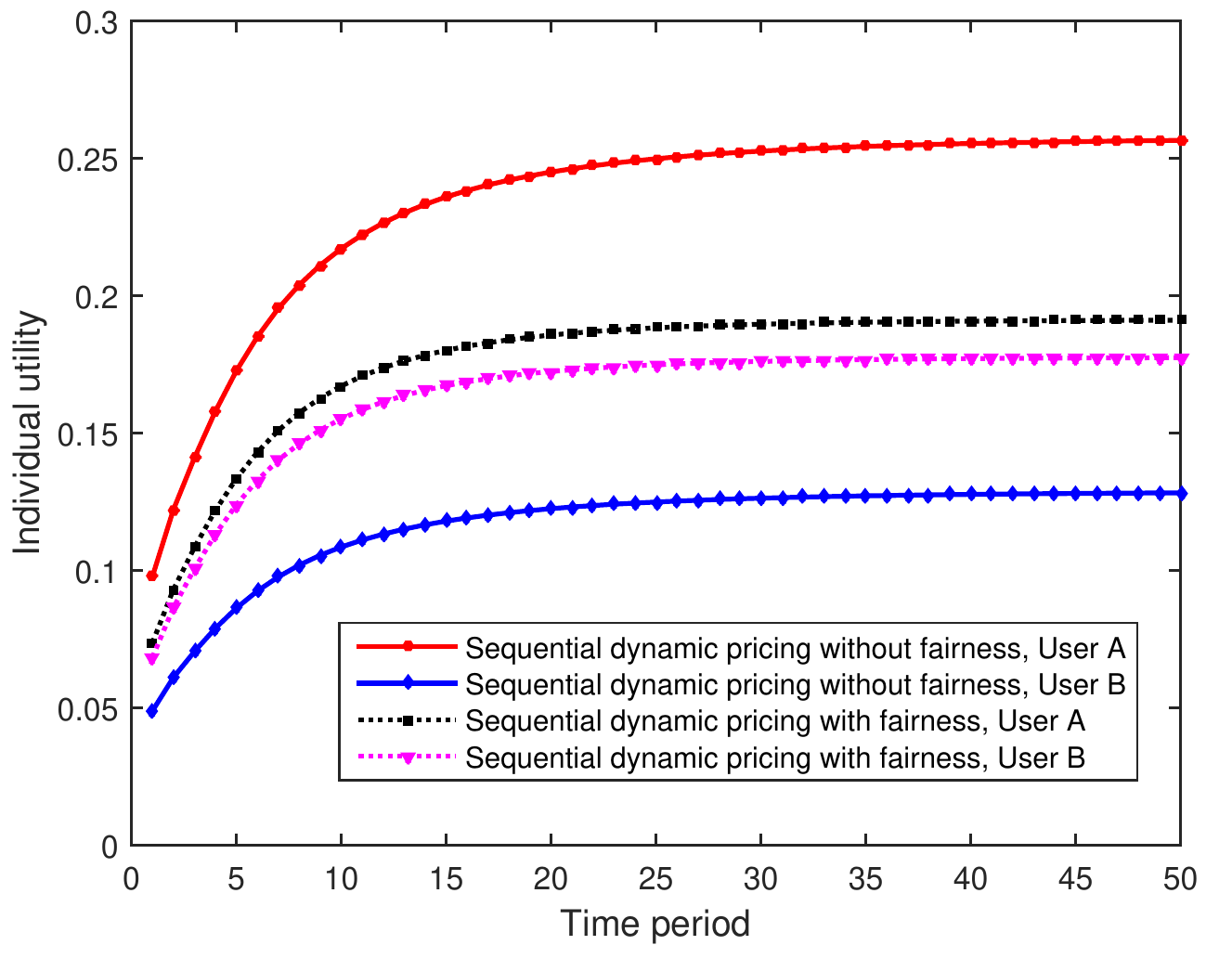}
\caption{The illustration of individual utility of users with and without social fairness consideration.}
\label{fig:side:c}
\end{minipage}
\vspace*{-6mm}
\end{figure}

\subsection{Sequential dynamic pricing}
From Figs.~\ref{fig:side:a}-\ref{fig:side:b} we observe that the convergence of sequential dynamic pricing (\textit{SeqDP}) in terms of the revenue of the MNO and the total utilities of users can be guaranteed within the first $40$ time periods. In addition, the convergence result of the proposed pricing outperforms the optimal value of the optimal static pricing (\textit{OSP}) scheme. Moreover, we compare the individual utility of two randomly selected users, under \textit{SeqDP} with and without social fairness consideration, as illustrated in Fig.~\ref{fig:side:c}. This demonstrates that the modified \textit{SeqDP} is able to achieve the social fairness in terms of individual network utility.

To evaluate the total utilities of mobile users gained from social data demand and the revenue of the MNO, we now compare the revenue and total utilities of the proposed \textit{SeqDP} and those of \textit{OSP} in Figs.~\ref{Fig:SeqDP_edgeprobability}-\ref{Fig:SeqDP_congestion}. As a benchmark, we also evaluate the performance when the social data demand of users is not interdependent. This is a special case of our proposed socially aware user utility where all the social ties equal 0. We also compare the performance under the ER based social graph model (social graph-ER) with the real dataset, i.e., Brightkite based social graph model (social graph-Brightkite).

\begin{figure}
\begin{minipage}[t]{0.5\linewidth}
\centering
\includegraphics[width=2.35in]{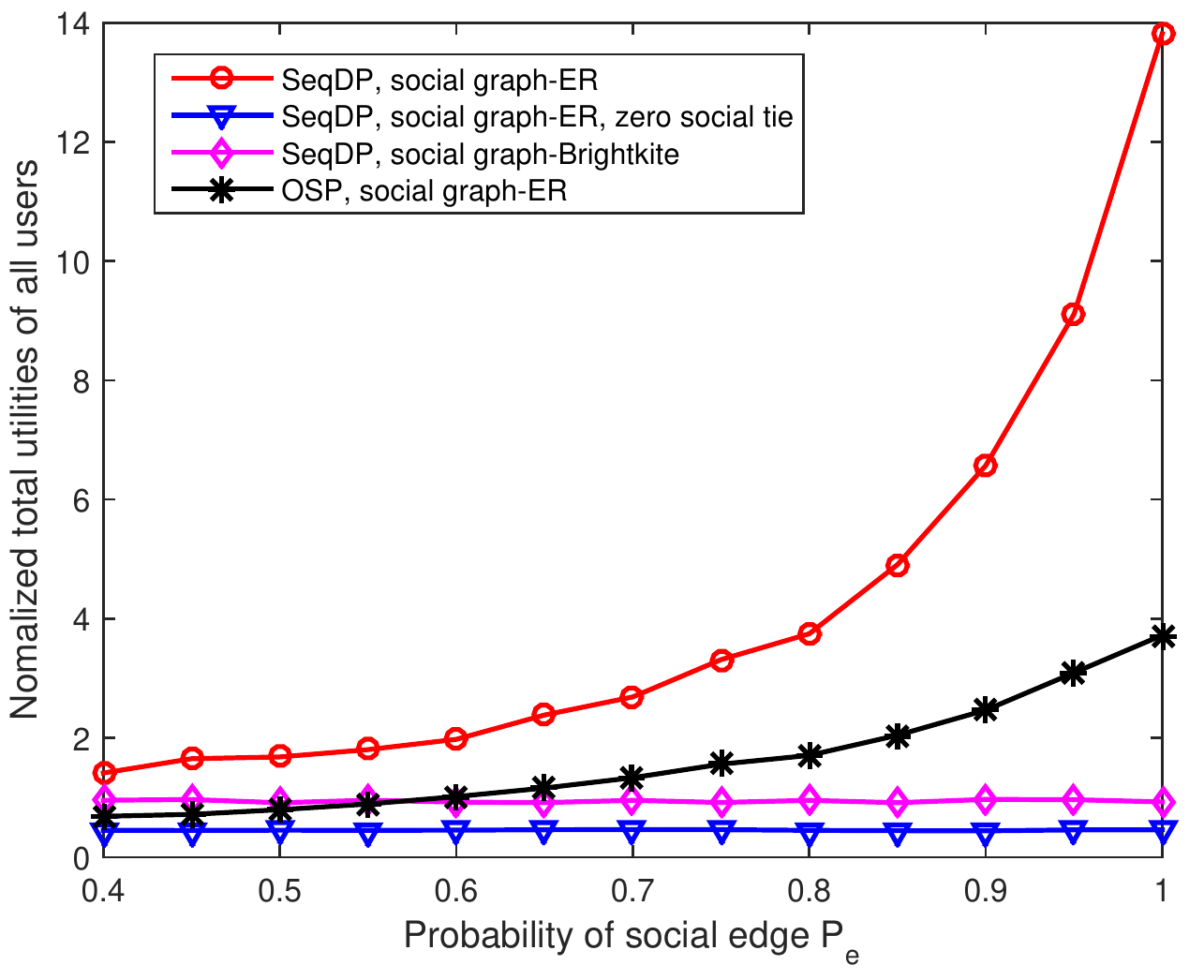}
\end{minipage}%
\begin{minipage}[t]{0.5\linewidth}
\centering
\includegraphics[width=2.35in]{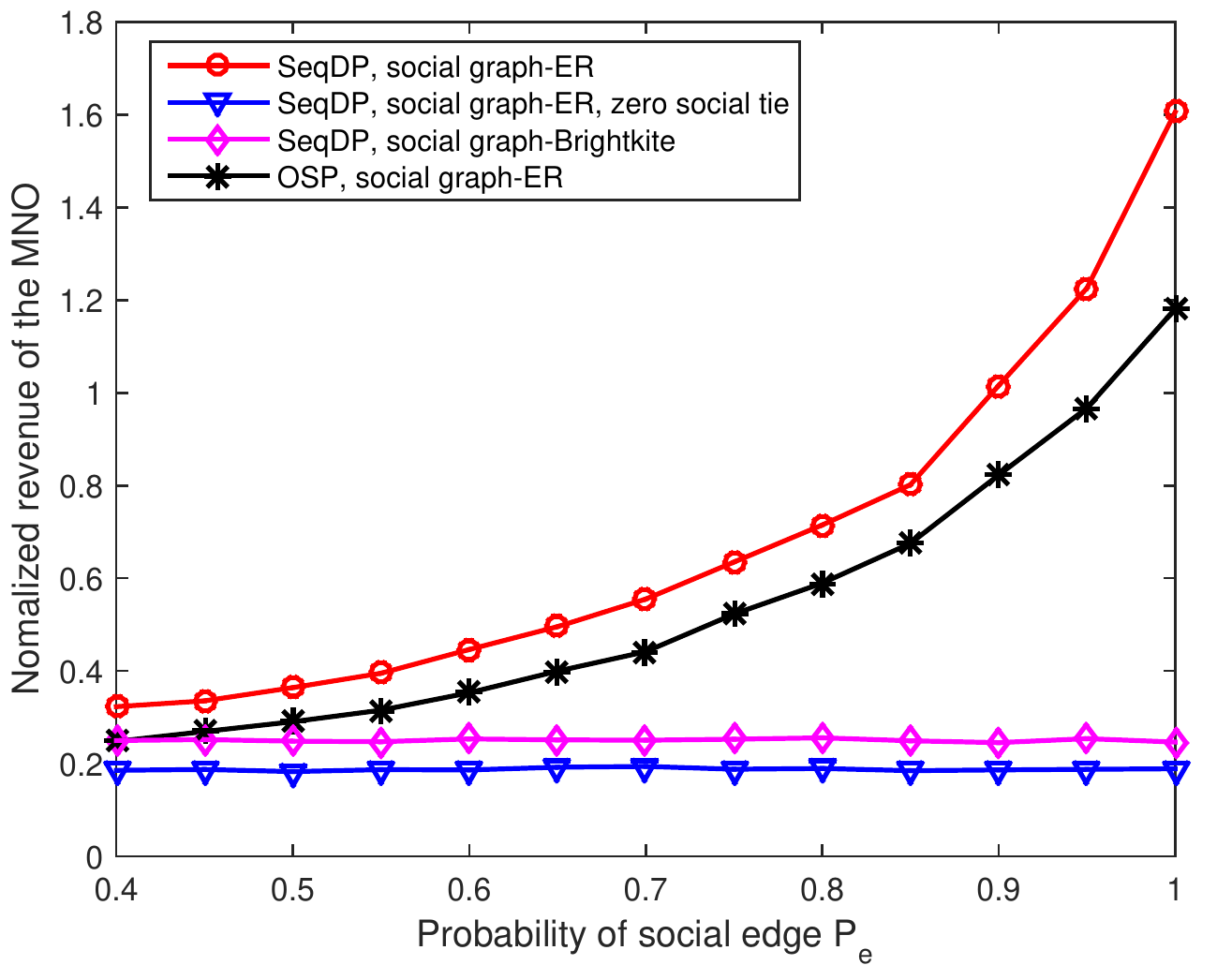}
\end{minipage}\caption{Normalized total utilities of users and normalized revenue of the MNO versus the probability of social edge.}\label{Fig:SeqDP_edgeprobability}
\vspace*{-4mm}
\end{figure}

\begin{figure}
\begin{minipage}[t]{0.5\linewidth}
\centering
\includegraphics[width=2.4in]{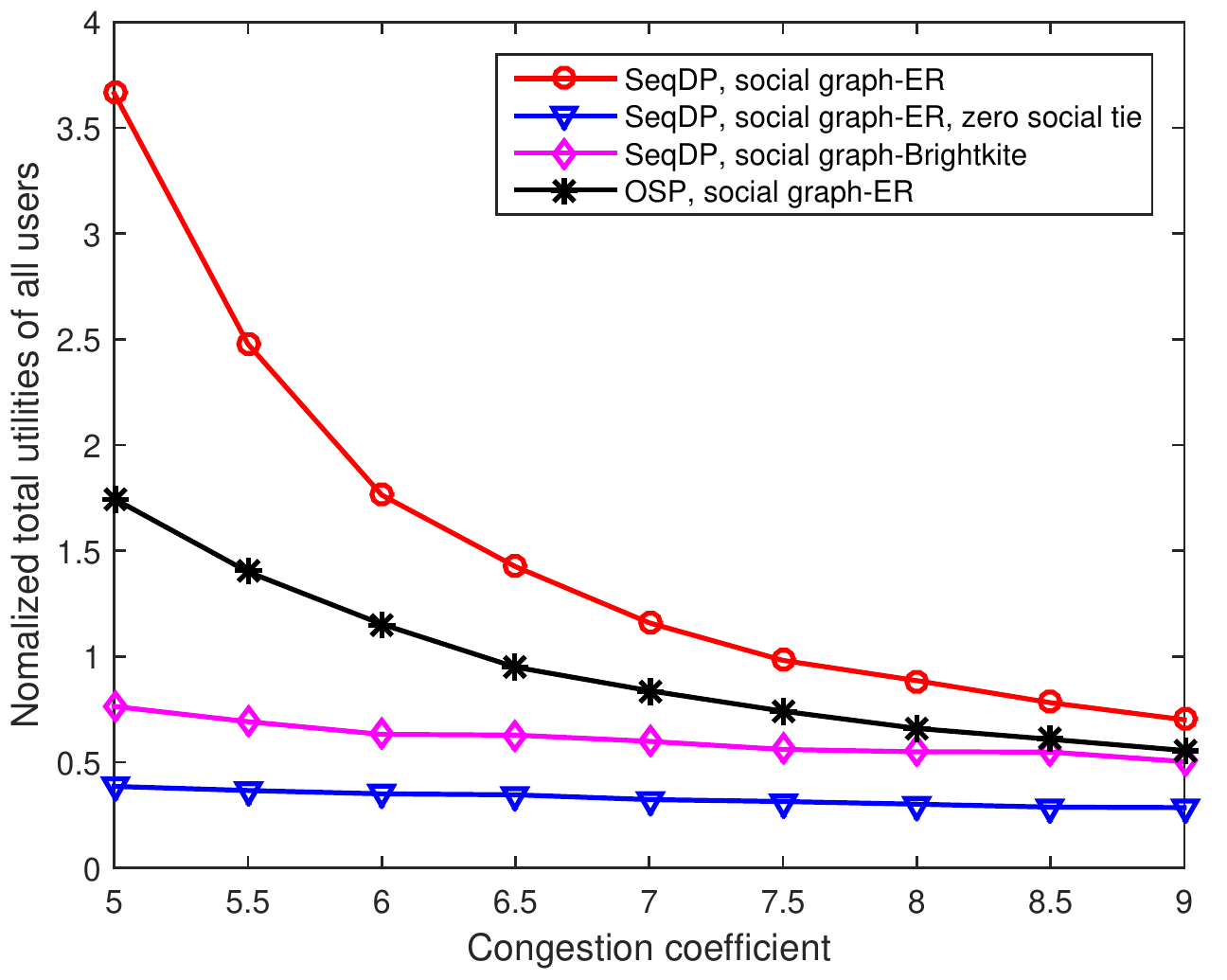}
\end{minipage}%
\begin{minipage}[t]{0.5\linewidth}
\centering
\includegraphics[width=2.4in]{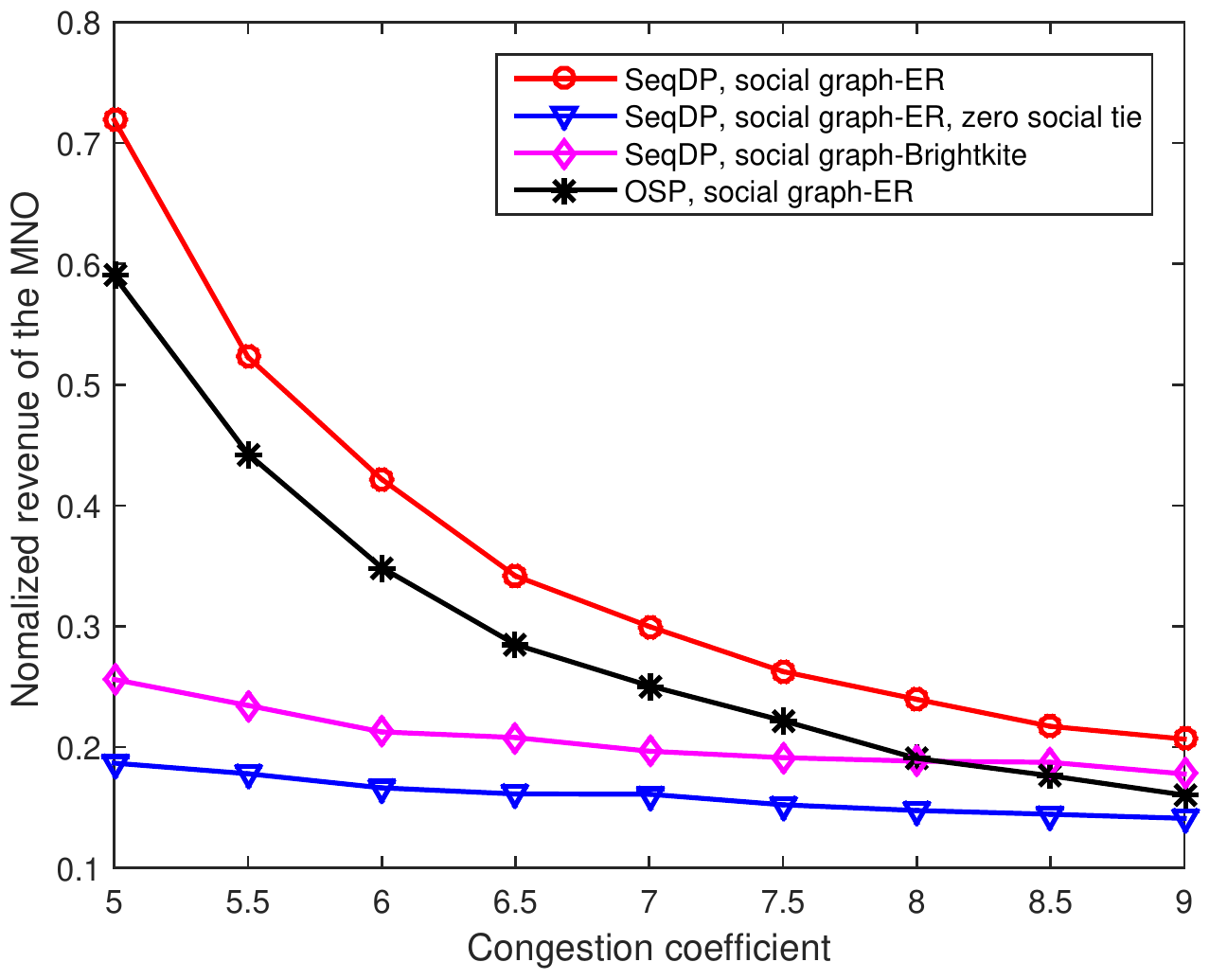}
\end{minipage}\caption{Normalized total utilities of users and normalized revenue of the MNO versus the congestion coefficient.}\label{Fig:SeqDP_congestion}
\vspace*{-6mm}
\end{figure}

From Fig.~\ref{Fig:SeqDP_edgeprobability}, the total utilities increase with the probability of social edge, and the total utilities achieved in the proposed \textit{SeqDP} are much larger than those of \textit{OSP} when the probability of social edge is higher. As the probability of social edge increases, the number of social neighbours of one user increases, and the additional benefits due to neighbours' social data demand become greater, and thus the total utilities increase. Therefore, the revenue of the MNO obtained by the \textit{SeqDP} increases with the increase of probability of social edge. The intuition is that the higher probability of social edge leads to higher social data demand due to underlying network effects, which in turn promotes the revenue of the MNO. This can be verified by the performance under the ER based social graph model with zero social tie. In this special case, there is no network effect, and thus the probability of social edge does not affect the performance in terms of the total utilities of users and the revenue of the MNO. Moreover, we evaluate the impact of congestion effects on the MNO and users, as shown in Fig.~\ref{Fig:SeqDP_congestion}. Under all the cases, we observe that the total utilities of users and the revenue of the MNO decrease as the congestion coefficient increases. Intuitively, with larger congestion, the negative impact coming from others' social data demand increases, and thus the utility of each user becomes lower. Consequently, the decreasing social data demand leads to the decrease of the revenue of the MNO. Furthermore, given the number of users, $50$, the probability of social edge in the real dataset Brightkite can be obtained from Fig.~\ref{Fig:RealData}, which is lower than 0.3. Thus, the performance of \textit{SeqDP} under ER based social graph is better than that under the Brightkite based social graph, as illustrated in Figs.~\ref{Fig:SeqDP_edgeprobability}-\ref{Fig:SeqDP_congestion}.

\subsection{Simultaneous dynamic pricing}
\begin{figure}[t]
\centering
\includegraphics[width=0.45\textwidth]{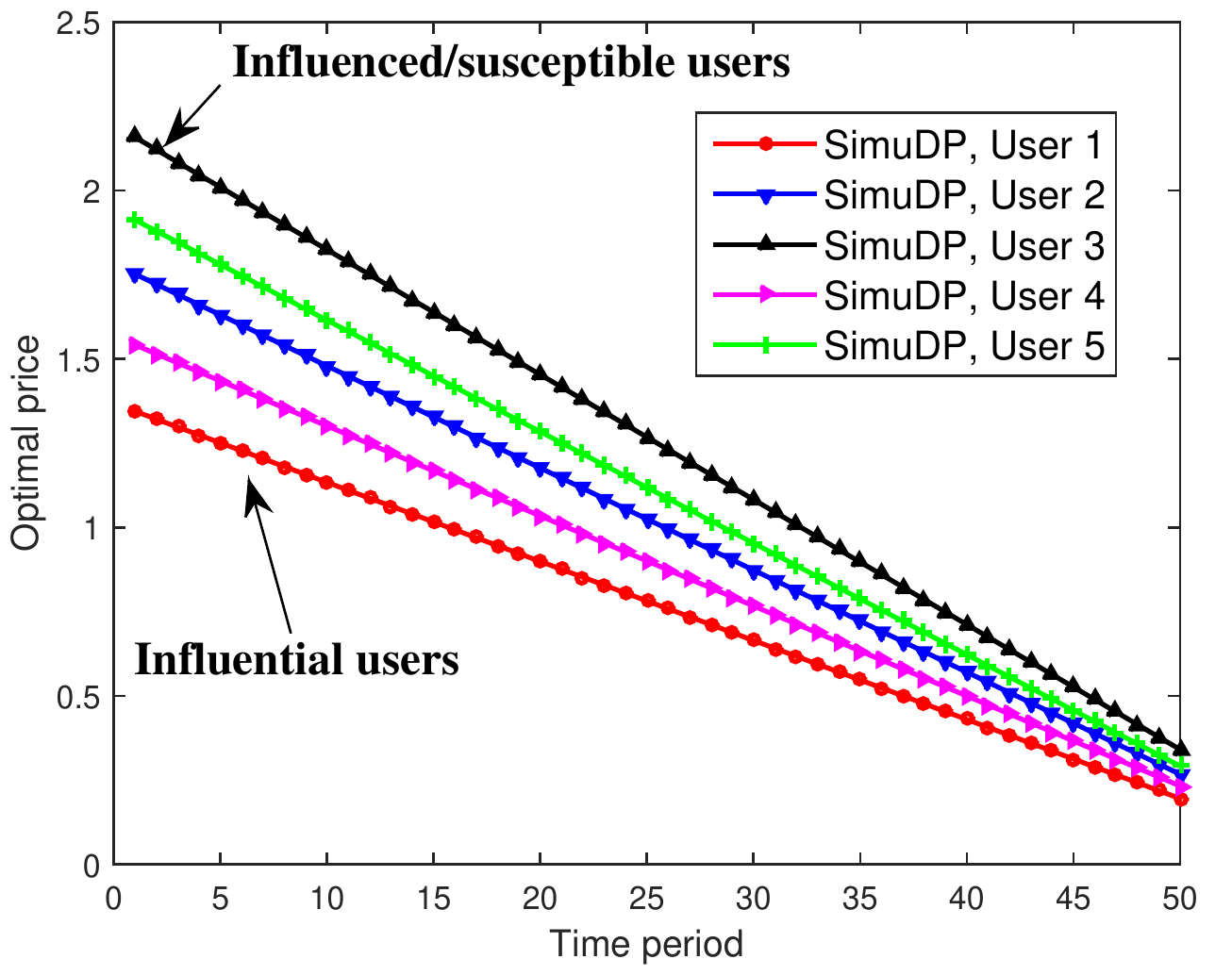}
\caption{The illustration of optimal prices of different users versus time periods under simultaneous dynamic pricing.}\label{Fig:SimuDP_Price}
\vspace*{-6mm}
\end{figure}

Figure~\ref{Fig:SimuDP_Price} illustrates the optimal price of selected 5 users in 50 time periods under \textit{SimuDP}. We observe that the optimal price decreases with the increasing time periods, which is consistent with Proposition 3. Furthermore, the decreasing slope is different for different users. Specifically, the selected 5 users in Fig.~\ref{Fig:SimuDP_Price} have different social relation factors. Therein, user 1 is the most influential user and user 3 is the least influential user, i.e., the most susceptible or influenced user. In other words, user 1 can influence more users due to network effects and user 3 is in the opposite. We observe that the price offered to the more influential user is lower and the decreasing rate is lower. The reason is that the MNO wants to offer the discount price to the influential users which can bring more potential users in subsequent time periods. However, the new coming users may lead to the decrease of user utility due to the congestion effects. Therefore, the decreasing rate of the price offered to the influential user is not higher than that of influenced or susceptible users.

\begin{figure}
\begin{minipage}[t]{0.5\linewidth}
\centering
\includegraphics[width=2.4in]{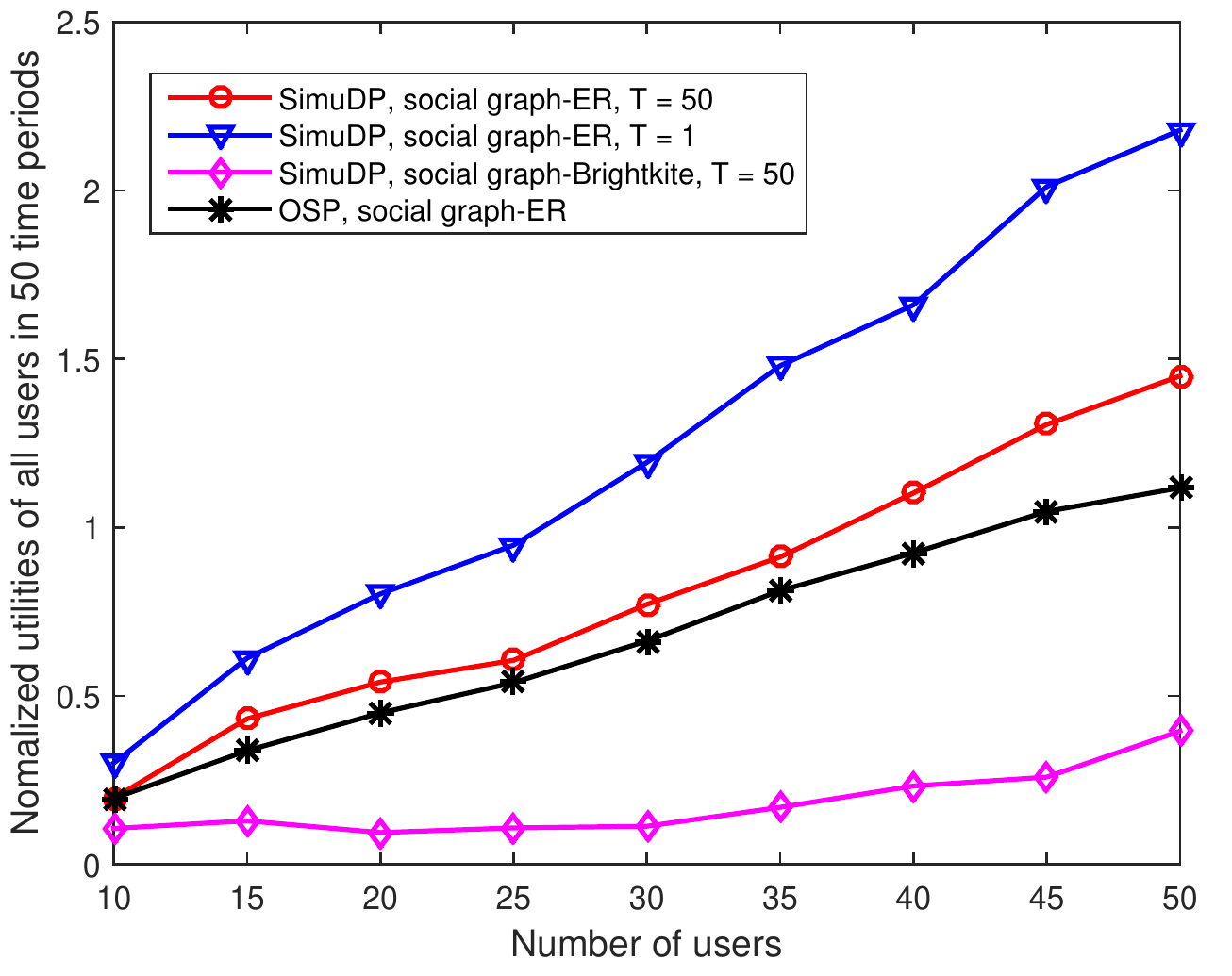}
\end{minipage}%
\begin{minipage}[t]{0.5\linewidth}
\centering
\includegraphics[width=2.4in]{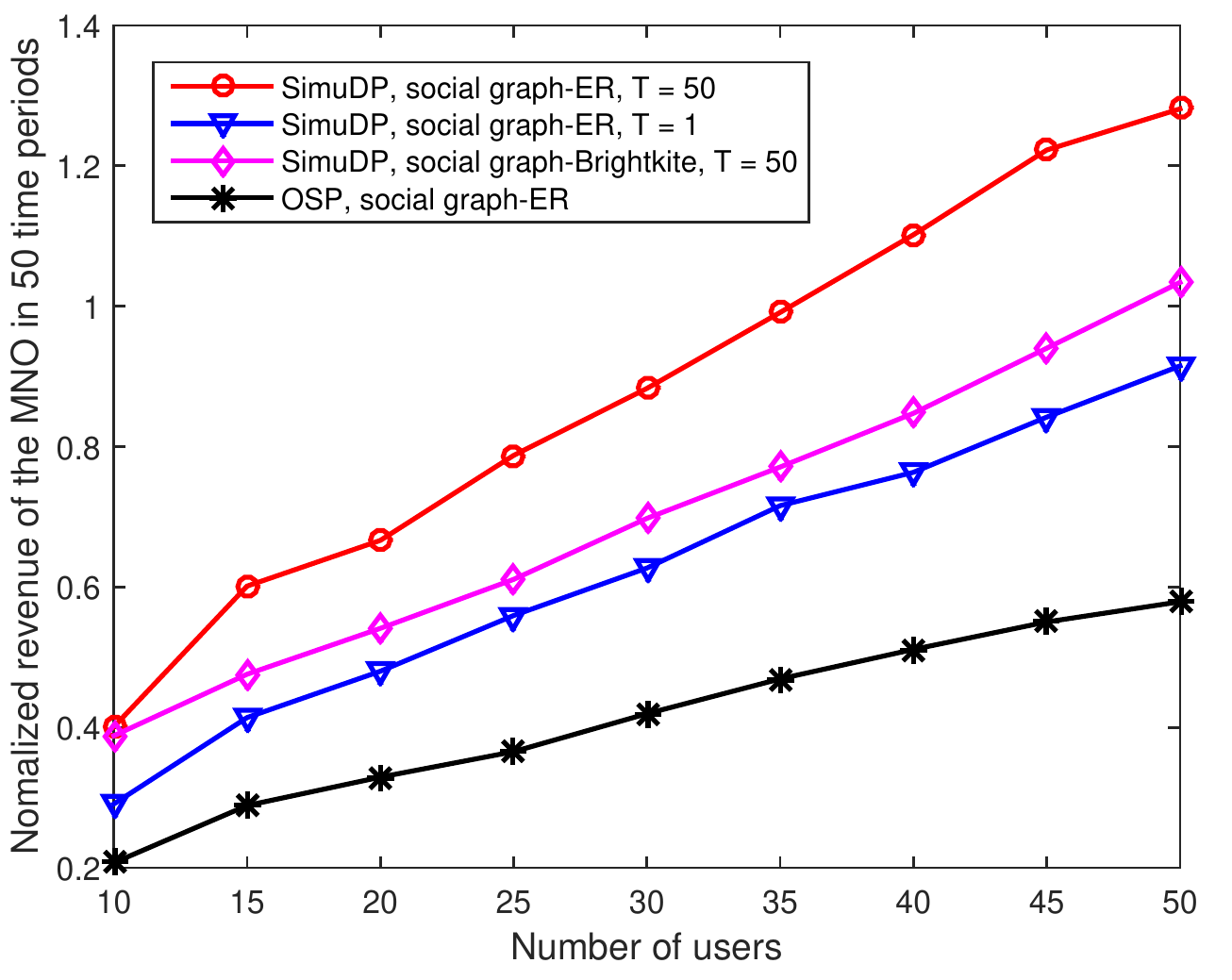}
\end{minipage}\caption{Normalized total utilities of users and normalized revenue of the MNO in 50 time periods versus the number of mobile users.}\label{Fig:SeqDP_number}
\vspace*{-6mm}
\end{figure}

We now fix the selling time periods as 50, and investigate the impact of different parameters on the performance of our proposed simultaneous dynamic pricing (\textit{SimuDP}) scheme, In addition to the \textit{OSP} and the real dataset Brightkite, we further consider the following two cases. In the first case, the MNO can only foresee 1 time period, i.e., the MNO maximizes its current revenue in each time period myopically, i.e., a greedy scheme. In the second case, the MNO can foresee 50 time periods, i.e., the MNO is fully rational and can maximize its revenue in the entire 50 time periods. Figure~\ref{Fig:SeqDP_number} illustrates the total utilities of users and the revenue of the MNO when the number of users increases. Under all the cases, we find that both the total utilities of users and the revenue of the MNO increase with the increase of the number of users. The reason is that adding more users would enhance each user's interactions with others, and thus potentially stimulate more social data demand of new coming users. Consequently, the increasing content demand results in a higher revenue of the MNO. Under the greedy \textit{SimuDP} where $T=1$, we observe that the revenue of the MNO is lower than that under the rational \textit{SimuDP} where $T=50$. The reason is that the fully rational MNO is able to foresee the social data demand in entire 50 time periods and thus extract more surplus with a higher revenue. Therefore, the total utilities of users from greedy \textit{SimuDP} are lower than that from the rational \textit{SimuDP}.

As expected, in the Brightkite based social graph, both the total utilities of users and the revenue of the MNO is lower than those from the ER based social graph. The intuitive reason is that the probability of social edge in the real dataset between any pair of users is smaller. Thus, when more users join, the social edge between the new users and existing users is weak, and accordingly the congestion effects dominate the network effects. Moreover, in Fig.~\ref{Fig:SeqDP_socialtie}, we observe that the performance of \textit{SimuDP} in terms of the total utilities of users and the revenue of the MNO increase when the average value of social tie increases. The reason is that as the network effects become stronger, the social data demand of each user is promoted due to stronger positive interdependency of each other. Consequently, the increase of social data demand results in a higher level of the revenue of the MNO. Both Figs.~\ref{Fig:SeqDP_number} and~\ref{Fig:SeqDP_socialtie} demonstrate the superior performance of the \textit{SimuDP} in terms of the revenue of the MNO over \textit{OSP}.

\begin{figure}
\begin{minipage}[t]{0.5\linewidth}
\centering
\includegraphics[width=2.5in]{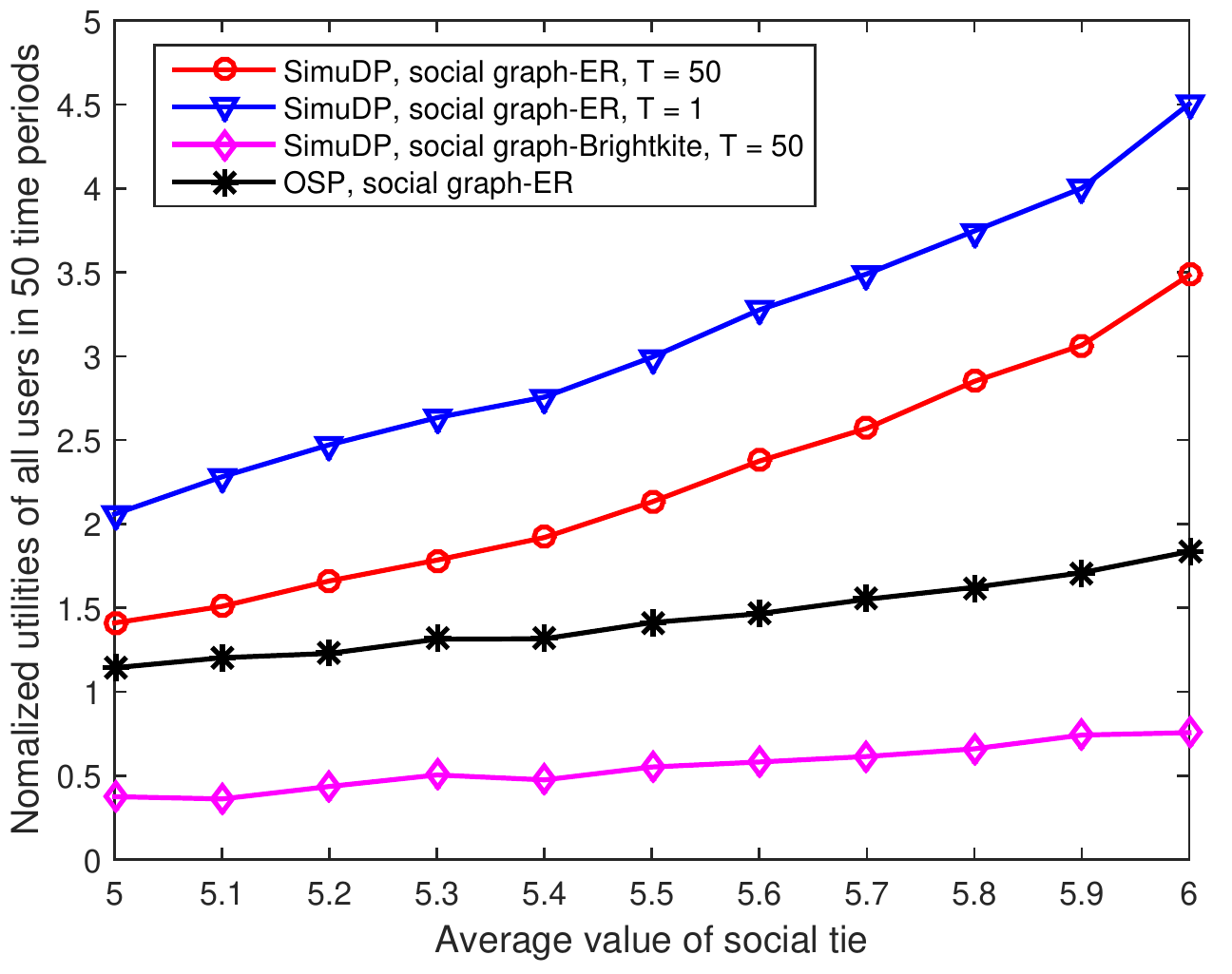}
\end{minipage}%
\begin{minipage}[t]{0.5\linewidth}
\centering
\includegraphics[width=2.5in]{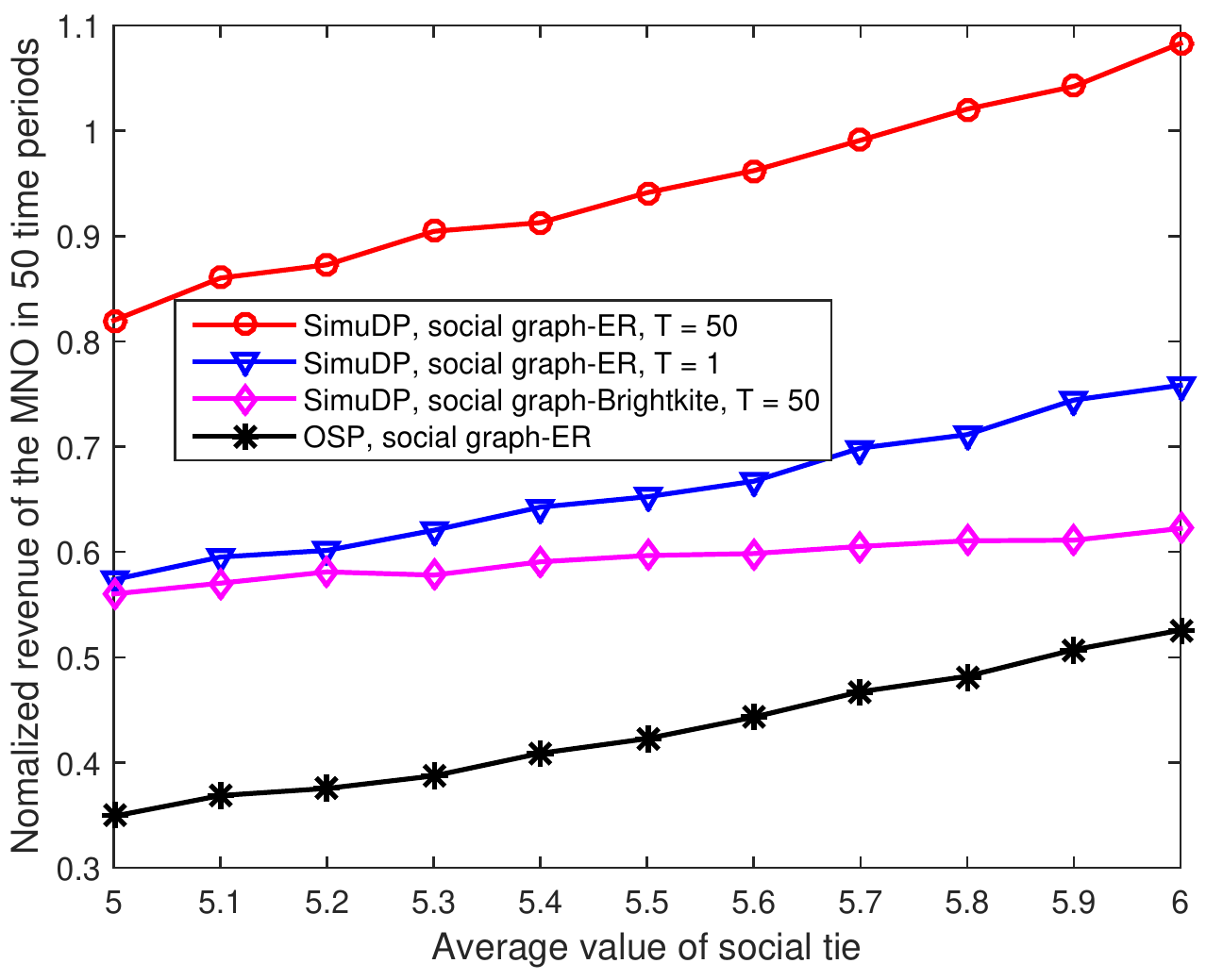}
\end{minipage}\caption{Normalized total utilities of users and normalized revenue of the MNO in 50 time periods versus average value of social tie.}\label{Fig:SeqDP_socialtie}
\vspace*{-6mm}
\end{figure}

\section{Conclusion}\label{Sec:Conclusion}

In this paper, we have presented a revenue maximization framework with the dynamic pricing schemes in mobile social data market. We have proposed a sequential dynamic pricing scheme where the mobile network operator individually offers a certain price to each mobile user for social data access in multiple time periods sequentially and repeatedly. The proposed pricing scheme has explored the network effects in the social domain and the congestion effects in the network domain. Furthermore, the alternative pricing policy has been implemented to ensure the social fairness among users with respect to individual utilities. 
Moreover, we have proposed a simultaneous dynamic pricing scheme to obtain more insights. 
Through the Erd\H{o}s-R\'enyi graph and the real dataset based social graph, we have conducted extensive performance evaluations to validate the outperformance of the dynamics of pricing schemes against static ones in terms of the revenue of the operator. 

\appendix
\subsection{Proof of Proposition 1}
\begin{proof}
We firstly denote $v$ as the eigenvector of ${(2{{\bf{\Lambda }}_c})^{ - 1}}\left( {{\bf{G}} - {\cal C}} \right)$ and $\lambda$ is the corresponding eigenvalue. Moreover, we denote $v_i$ as the largest entry of $v$ in absolute value, $\forall j \in \cal N$. Therefore, $\left| \lambda  \right|$ is not unbounded, and the proof steps are shown as follows:
\begin{eqnarray}
\left| {\lambda {v_i}} \right| &=& \left| {{{\left[ {{{(2{{\bf{\Lambda }}_c})}^{ - 1}}\left( {{\bf{G}} - {\cal C}} \right)} \right]}_i}v} \right| \nonumber\\
&\le& \sum\limits_{j \in \cal N} {{{\left[ {{{(2{{\bf{\Lambda }}_c})}^{ - 1}}\left( {{\bf{G}} - {\cal C}} \right)} \right]}_{ij}}\left| {{v_j}} \right|}
\le \frac{{\left| {{v_i}} \right|}}{{4{b_i} + c}}\sum\limits_{j \in \cal N} {\left( {{g_{ij}} - c} \right)}  <\frac{{\left| {{v_i}} \right|}}{2},
\end{eqnarray}
where ${{{\left[ {{{(2{{\bf{\Lambda }}_c})}^{ - 1}}\left( {{\bf{G}} - {\cal C}} \right)} \right]}_i}}$ represents the $i$th row of ${{{(2{{\bf{\Lambda }}_c})}^{ - 1}}\left( {{\bf{G}} - {\cal C}} \right)}$. In particular, the former two inequalities are based on the fact that
\begin{equation}
{{{\left[ {{{(2{{\bf{\Lambda }}_c})}^{ - 1}}\left( {{\bf{G}} - {\cal C}} \right)} \right]}_{ij}}}= \frac{{{g_{ij}} - c}}{{4{b_i} + c}}>0,
\end{equation}
and the last inequality is satisfied under Assumption 1. Therefore, we can make sure that every eigenvalue of ${(2{{\bf{\Lambda }}_c})^{ - 1}} ({{\bf{G}} - {\cal C}})$ is smaller than 1.

We observe that each eigenvalue of ${\bf{I}}-{(2{{\bf{\Lambda }}_c})^{ - 1}}\left( {{\bf{G}} - {\cal C}} \right)$ is $1-\lambda$, where $\lambda$ is an eigenvalue of ${(2{{\bf{\Lambda }}_c})^{ - 1}}\left( {{\bf{G}} - {\cal C}} \right)$. Since every eigenvalue of ${(2{{\bf{\Lambda }}_c})^{ - 1}}\left( {{\bf{G}} - {\cal C}} \right)$ is smaller than 1, we conclude that none of the eigenvalues of ${\bf{I}}-{(2{{\bf{\Lambda }}_c})^{ - 1}}\left( {{\bf{G}} - {\cal C}} \right)$ is zero, and thus this matrix is invertible. This also indicates that the matrix $(2{{\bf{\Lambda }}_c} - {\bf{G}} + {\mathcal{C}})$ is invertible. Moreover, we have
\begin{eqnarray}
{(2{{\bf{\Lambda }}_c} - {\bf{G}} + {\cal C})^{ - 1}} &=& {\left[ {{\bf{I}} - {{({2}{{\bf{\Lambda }}_{\bf{c}}})}^{ - {{1}}}}\left( {{\bf{G}} - {\mathcal{C}}} \right)} \right]^{ - 1}}{({2}{{\bf{\Lambda }}_{\bf{c}}})^{ - {{1}}}}\nonumber\\
 &=& \sum\limits_{j = 1}^\infty  {{{\left[ {{{({2}{{\bf{\Lambda }}_{\bf{c}}})}^{ - {{1}}}}\left( {{\bf{G}} - {\mathcal{C}}} \right)} \right]}^j}{{({2}{{\bf{\Lambda }}_{\bf{c}}})}^{ - {{1}}}}}.
\end{eqnarray}
This is based on the fact that the spectral radius of ${(2{{\bf{\Lambda }}_c})^{ - 1}}\left( {{\bf{G}} - {\cal C}} \right)$ is smaller than $1$. Since we know that the entries of ${{\bf{\Lambda }}_c}$ and ${{\bf{G}} - {\cal C}}$ are both non-negative, accordingly we can conclude that all the entries of ${(2{{\bf{\Lambda }}_c} - {\bf{G}} + {\cal C})^{ - 1}}$ are non-negative. The proof is now completed.
\end{proof}

\subsection{Proof of Proposition 2}
\begin{proof}
To prove that the spectrum radius of $[(2{{\bf{\Lambda }}_{\bf{c}}}-{\bf{\Lambda }})(2{{\bf{\Lambda }}_{\bf{c}}} - {\bf{G}} + {\mathcal{C}})^{-1}]^2$ is smaller than 1, we only need to prove that the spectrum radius of $(2{{\bf{\Lambda }}_{\bf{c}}} - {\bf{\Lambda }}){(2{{\bf{\Lambda }}_{\bf{c}}} - {\bf{G}} + {\cal C})^{ - 1}}$ is smaller than 1. We have
\begin{equation}
(2{{\bf{\Lambda }}_{\bf{c}}} - {\bf{\Lambda }}){(2{{\bf{\Lambda }}_{\bf{c}}} - {\bf{G}} + {\cal C})^{ - 1}}= (2{{\bf{\Lambda }}_{\bf{c}}} - {\bf{\Lambda }}){\left( {2{{\bf{\Lambda }}_{\bf{c}}}} \right)^{ - 1}}{\left( {\bf{I} - {{\left( {2{{\bf{\Lambda }}_{\bf{c}}}} \right)}^{ - 1}}\left( {{\bf{G}} + {\cal C}} \right)} \right)^{ - 1}}.
\end{equation}
The spectrum radius of ${\left( {{\bf{I}} - {{\left( {2{{\bf{\Lambda }}_{\bf{c}}}} \right)}^{ - 1}}\left( {{\bf{G}} + {\cal C}} \right)} \right)^{ - 1}}$ is smaller than $1$, which has been shown in the proof of Proposition 1. Then, the result of the first two terms is
\begin{eqnarray}
(2{{\bf{\Lambda }}_{\bf{c}}} - {\bf{\Lambda }}){\left( {2{{\bf{\Lambda }}_{\bf{c}}}} \right)^{ - 1}}&=&diag\left( {{{\left[ {\frac{{4{b_i} + c - 2{b_i}}}{{4{b_i} + c}}} \right]}^\top}} \right) \nonumber \\ &=&diag\left( {{{\left[ {\frac{{2{b_i} + c}}{{4{b_i} + c}}} \right]}^\top}} \right)<\bf{I}.
\end{eqnarray}
This indicates that the spectrum radius of the result for the first two terms is also smaller than 1. The proof is now completed.
\end{proof}

\subsection{Proof of Theorem 3}
\begin{proof}
Following~(\ref{Eq:8}) from Theorem 1, the revenue of the MNO ${\Pi ^{(k)}}$ is expressed by
\begin{eqnarray}
{\Pi ^{(k)}} &=& \left[{{\bf{a}}^{(k)}} - ({\bf{\Lambda }} - {\bf{G}} + {\mathcal{C}}){{\bf{x}}^{(k)}}\right]^\top{{\bf{x}}^{(k)}}\nonumber \\
 &=& \left[{{\bf{a}}^{(k)}} - ({\bf{\Lambda }} - {\bf{G}} + {\mathcal{C}}){(2{{\bf{\Lambda }}_c} - {\bf{G}} + {\mathcal{C}})^{ - 1}}{{\bf{a}}^{(k)}}\right]^\top {(2{{\bf{\Lambda }}_c} - {\bf{G}} + {\mathcal{C}})^{ - 1}}{{\bf{a}}^{(k)}}\nonumber \\
 &=& \left({{\bf{a}}^{(k)}}\right)^\top(2{{\bf{\Lambda }}_c} - {\bf{\Lambda }}){(2{{\bf{\Lambda }}_c} - {\bf{G}} + {\mathcal{C}})^{ - 2}}{{\bf{a}}^{(k)}}.
\end{eqnarray}
Then, according to~(\ref{Eq:5}), the revenue under our proposed pricing, ${\Pi _d}$, is given by
\begin{eqnarray}
&{\Pi _d}& = \sum\limits_{k = 1}^\infty  {{\Pi ^{(k)}}} = \sum\limits_{k = 1}^\infty  {\left({{\bf{a}}^{(k)}}\right)^\top(2{{\bf{\Lambda }}_c} - {\bf{\Lambda }}){{(2{{\bf{\Lambda }}_c} - {\bf{G}} + {\mathcal{C}})}^{ - 2}}{{\bf{a}}^{(k)}}}\nonumber \\
  &=& \sum\limits_{k = 1}^\infty  {\left({{\bf{a}}^{(k)}}\right)^\top\Big[{\bf{I}} - ({\bf{\Lambda }} - {\bf{G}} + {\mathcal{C}}){{(2{{\bf{\Lambda }}_c} - {\bf{G}} + {\mathcal{C}})}^{ - 1}}} {\Big]^{k - 1}}(2{{\bf{\Lambda }}_c} - {\bf{\Lambda}})\nonumber \\ &\times&
 {(2{{\bf{\Lambda }}_c} - {\bf{G}} + {\mathcal{C}})^{ - 2}}[{\bf{I}} - ({\bf{\Lambda }} - {\bf{G}} + {\mathcal{C}}) {(2{{\bf{\Lambda }}_c} - {\bf{G}} + {\mathcal{C}})^{ - 1}}]^{k - 1}{\bf{a}}\nonumber \\
  &=& \sum\limits_{k = 1}^\infty  \left({{\bf{a}}^{(k)}}\right)^\top{{\Big[(2{{\bf{\Lambda }}_c} - {\bf{\Lambda }}){{(2{{\bf{\Lambda }}_c} - {\bf{G}} + {\mathcal{C}})}^{ - 1}}\Big]}^{k - 1}}(2{{\bf{\Lambda }}_c} -  {\bf{\Lambda }})\nonumber \\ &\times&
  {{(2{{\bf{\Lambda }}_c} - {\bf{G}} + {\mathcal{C}})}^{-2}}{{\Big[{{(2{{\bf{\Lambda }}_c} - {\bf{G}} + {\mathcal{C}})}^{ - 1}}\Big]}^{k - 1}}{\bf{a}}\nonumber \\
 &=& \sum\limits_{k = 1}^\infty  {\left({{\bf{a}}^{(k)}}\right)^\top{{(2{{\bf{\Lambda }}_c} - {\bf{\Lambda }})}^{ - 1}}{{\Big[(2{{\bf{\Lambda }}_c} - {\bf{\Lambda }}){{(2{{\bf{\Lambda }}_c} - {\bf{G}} + {\mathcal{C}})}^{ - 1}}\Big]}^{2k}}{\bf{a}}} \nonumber \\
 &=&  {{\bf{a}}^\top{{(2{{\bf{\Lambda }}_c} - {\bf{\Lambda }})}^{ - 1}}{{\left[(2{{\bf{\Lambda }}_c} - {\bf{\Lambda }}){{(2{{\bf{\Lambda }}_c} - {\bf{G}} + {\mathcal{C}})}^{ - 1}}\right]}^2}} {\left\{ {\bf{I}} - {\left[({\bf{\Lambda }} - {\bf{G}} + {\mathcal{C}}){(2{{\bf{\Lambda }}_c} - {\bf{G}} + {\mathcal{C}})^{ - 1}}\right]^2}\right\} ^{ - 1}}{\bf{a}} \nonumber \\
 &=&   {{\bf{a}}^\top{{(2{{\bf{\Lambda }}_c} - {\bf{\Lambda }})}^{ - 1}}{{\Big[(2{{\bf{\Lambda }}_c} - {\bf{\Lambda }}){{(2{{\bf{\Lambda }}_c} - {\bf{G}} + {\mathcal{C}})}^{ - 1}}\Big]}^2}}  \nonumber \\ &\times&  ({\bf{\Lambda }} - {\bf{G}} + {\mathcal{C}}){(2{{\bf{\Lambda }}_c} - {\bf{G}} + {\mathcal{C}})^{ - 1}}\Big[2{\bf{I}} - ({\bf{\Lambda }} - {\bf{G}} + {\mathcal{C}}) {(2{{\bf{\Lambda }}_c} - {\bf{G}} + {\mathcal{C}})^{-1}}\Big] ^{ - 1}{\bf{a}}\nonumber \\  &=&   {{\bf{a}}^\top(2{{\bf{\Lambda }}_c} - {\bf{\Lambda }})
 (2{{\bf{\Lambda }}_c} - {\bf{G}} + {\mathcal{C}})^{ - 1}} {({\bf{\Lambda }} - {\bf{G}} + {\mathcal{C}})^{ - 1}}
 {\Big[2{\bf{I}} - ({\bf{\Lambda }} - {\bf{G}} + {\mathcal{C}}){(2{{\bf{\Lambda }}_c} - {\bf{G}} + {\mathcal{C}})^{ - 1}}\Big]^{ - 1}}{\bf{a}}.\nonumber \\
\end{eqnarray}
According to~\cite{zhang2016social}, the revenue under optimal static pricing ${\Pi}_s$ is expressed by
\begin{eqnarray}
{\Pi _s}&=& \sum\limits_{i \in \cal N} {\widehat {{p}}_i\widehat {{x}}_i} \nonumber \\
 &=& \left[{\bf{a}} - ({\bf{\Lambda }} - {\bf{G}} + {\cal C}){\bf{x}}\right]{(2{{\bf{\Lambda }}_c} - {\bf{G}} + {\cal C})^{ - 1}}{\bf{a}}\nonumber\\
 &=& \left[{\bf{a}} - ({\bf{\Lambda }} - {\bf{G}} + {\cal C}){(2{{\bf{\Lambda }}_c} - {\bf{G}} + {\cal C})^{ - 1}}\right]{\bf{a}}{(2{{\bf{\Lambda }}_c} - {\bf{G}} + {\cal C})^{ - 1}}{\bf{a}}\nonumber\\
&=& {\widehat {\bf{p}}}^\top {\widehat {\bf{x}}}={\bf{a}}^\top(2{{\bf{\Lambda }}_c} - {\bf{\Lambda }}){\left[{(2{{\bf{\Lambda }}_c} - {\bf{G}} + {\mathcal{C}})^{ - 1}}\right]^2}{\bf{a}}.
\end{eqnarray}
Then, we prove that ${\Pi _d} - {\Pi _s}$ is non-negative according to~(\ref{Eq:11}) with Proposition 2.
\begin{eqnarray}\label{Eq:11}
&{\Pi _d}-{\Pi _s}& = {\bf{a}}^\top(2{{\bf{\Lambda }}_c} - {\bf{\Lambda }}){(2{{\bf{\Lambda }}_c} - {\bf{G}} + {\mathcal{C}})^{ - 1}}{({\bf{\Lambda }} - {\bf{G}} + {\mathcal{C}})^{ - 1}}\nonumber\\ &\times& {\Big[2{\bf{I}} - ({\bf{\Lambda }} - {\bf{G}} + {\mathcal{C}}){(2{{\bf{\Lambda }}_c} - {\bf{G}} + {\mathcal{C}})^{ - 1}}\Big]^{ - 1}}{\bf{a}} - {\bf{a}}(2{{\bf{\Lambda }}_c} - {\bf{\Lambda }}){\Big[{(2{{\bf{\Lambda }}_c} - {\bf{G}} + {\mathcal{C}})^{ - 1}}\Big]^2}{\bf{a}} \nonumber\\
&=& {\bf{a}}^\top(2{{\bf{\Lambda }}_c} - {\bf{\Lambda }}){(2{{\bf{\Lambda }}_c} - {\bf{G}} + {\mathcal{C}})^{ - 1}}\bigg\{ 2{\bf{I}} - {({\bf{\Lambda }} - {\bf{G}} + {\mathcal{C}})^{ - 1}} \nonumber\\ &\times& {\Big[2{\bf{I}} - ({\bf{\Lambda }} - {\bf{G}} + {\mathcal{C}}){(2{{\bf{\Lambda }}_c} - {\bf{G}} + {\mathcal{C}})^{ - 1}}\Big]^{ - 1}} - (2{{\bf{\Lambda }}_c} - {\bf{G}} + {\mathcal{C}})^{-1}\bigg\} {\bf{a}} \nonumber\\ &=& {\bf{a}}^\top(2{{\bf{\Lambda }}_c} - {\bf{\Lambda }}){(2{{\bf{\Lambda }}_c} - {\bf{G}} + {\mathcal{C}})^{ - 1}}{({\bf{\Lambda }} - {\bf{G}} + {\mathcal{C}})^{ - 1}}\nonumber\\ &\times& \Bigg\{{\Big[2{\bf{I}} - ({\bf{\Lambda }} - {\bf{G}} + {\mathcal{C}}){(2{{\bf{\Lambda }}_c} - {\bf{G}} + {\mathcal{C}})^{ - 1}}\Big]^{ - 1}} - ({\bf{\Lambda }} - {\bf{G}} + {\mathcal{C}}){(2{{\bf{\Lambda }}_c} - {\bf{G}} + {\mathcal{C}})^{ - 1}}\Bigg\} {\bf{a}}\nonumber\\&=& \underbrace {{\bf{a}}^\top(2{{\bf{\Lambda }}_c} - {\bf{\Lambda }}){(2{{\bf{\Lambda }}_c} - {\bf{G}} + {\cal C})^{ - 1}}{({\bf{\Lambda }} - {\bf{G}} + {\cal C})^{ - 1}}}_{\ge 0} \nonumber\\
&\times& \underbrace {{{\Big[2{\bf{I}} - ({\bf{\Lambda }} - {\bf{G}} + {\cal C}){{(2{{\bf{\Lambda }}_c} - {\bf{G}} + {\cal C})}^{ - 1}}\Big]}^{ - 1}}}_{ \ge 0}\underbrace {{{\Big[({\bf{\Lambda }} - {\bf{G}} + {\cal C}){{(2{{\bf{\Lambda }}_c} - {\bf{G}} + {\cal C})}^{ - 1}} - {\bf{I}}\Big]}^2}}_{ > 0} \ge 0.\nonumber\\
\end{eqnarray}

In addition, we compare the total utilities under our proposed sequential dynamic pricing with that under optimal static pricing~\cite{zhang2016social}. Let $\bf{y}$ and $\bf{x}$ be the optimal social data demand in proposed dynamic pricing and the optimal static pricing, respectively, and $\mathscr{U}_d$ in~(\ref{Eq:10}) is expressed by
\begin{eqnarray}
{{\mathscr U}_d}&=& {{\bf{a}}^\top}{\bf{y}} - {{\bf{y}}^\top}\frac{{\bf{\Lambda }}}{2}{\bf{y}} + {{\bf{y}}^\top}{\bf{G}}{\bf{y}} - {{\bf{y}}^\top}\frac{{{\mathcal{C}}}}{2}{\bf{y}} - {\Pi _d}\nonumber\\
 &=& {{\bf{a}}^\top}{\bf{y}} - {{\bf{y}}^\top}\left(\frac{{\bf{\Lambda }}}{2} + {\bf{G}} - \frac{{{\mathcal{C}}}}{2}\right){\bf{y}} - {\Pi _d}\nonumber\\
 &=& {{\bf{a}}^\top}{({\bf{\Lambda }} - {\bf{G}} + {\mathcal{C}})^{ - 1}}{\bf{a}} - {{\bf{a}}^\top}{({\bf{\Lambda }} - {\bf{G}} + {\mathcal{C}})^{ - 1}}\left(\frac{{\bf{\Lambda }}}{2} + {\bf{G}} - \frac{{{\mathcal{C}}}}{2}\right) {({\bf{\Lambda }} - {\bf{G}} + {\mathcal{C}})^{ - 1}}{\bf{a}} - {\Pi _d}\nonumber\\
 &=& {{\bf{a}}^\top}{({\bf{\Lambda }} - {\bf{G}} + {\mathcal{C}})^{ - 1}}\frac{{{\bf{\Lambda }} + {\mathcal{C}}}}{2}{({\bf{\Lambda }} - {\bf{G}} + {\mathcal{C}})^{ - 1}}{\bf{a}}  - {\Pi _d}.
\end{eqnarray}
From~\cite{zhang2016social}, we obtain $\mathscr{U}_s$, which is shown as:
\begin{eqnarray}
{{\mathscr U}_s} &=& {{\bf{a}}^\top}{\bf{x}} - {{\bf{x}}^\top}(\frac{{\bf{\Lambda }}}{2} + {\bf{G}} - \frac{{{\mathcal{C}}}}{2}){\bf{x}} - {\Pi _s}\nonumber\\
 &=& {{\bf{a}}^\top}{(2{{\bf{\Lambda }}_c} - {\bf{G}} + {\mathcal{C}})^{ - 1}}\left[{\bf{I}} - \left(\frac{{\bf{\Lambda }}}{2} + {\bf{G}} - \frac{{{\mathcal{C}}}}{2}\right)
 {(2{{\bf{\Lambda }}_c} - {\bf{G}} + {\mathcal{C}})^{ - 1}}\right]{\bf{a}} - {\Pi _s}\nonumber\\
 &=& {{\bf{a}}^\top}{(2{{\bf{\Lambda }}_c} - {\bf{G}} + {\mathcal{C}})^{ - 1}}\Big(2{{\bf{\Lambda }}_c} - \frac{{\bf{\Lambda }}}{2} + \frac{{{\mathcal{C}}}}{2}\Big)
 {\left(2{{\bf{\Lambda }}_c} - {\bf{G}} + {\mathcal{C}}\right)^{ - 1}}{\bf{a}} - {\Pi _s}.
\end{eqnarray}
Then, we let $\bf{z}$ be ${{\bf{a}}^\top}{({\bf{\Lambda }} - {\bf{G}} + {\mathcal{C}})^{ - 1}}{(2{{\bf{\Lambda }}_c} - {\bf{G}} + {\mathcal{C}})^{ - 1}}$. Therefore, we prove that ${{\mathscr U}_s} - {{\mathscr  U}_s}$ is non-negative based on~(\ref{Eq:12}), as shown as follows:
\begin{eqnarray}\label{Eq:12}
&&{{\mathscr U}_d}-{{\mathscr U}_s} = {\bf{z}}\frac{{{\bf{\Lambda }} + {\mathcal{C}}}}{2}(2{{\bf{\Lambda }}_c} + {\bf{\Lambda }} - 2{\bf{G}} + 2{\mathcal{C}})(2{{\bf{\Lambda }}_c} - {\bf{\Lambda }}){\bf{z}}^\top - {\bf{z}}(2{{\bf{\Lambda }}_c} - {\bf{\Lambda }})\nonumber\\ &\times&
{\left[2{\bf{I}} - ({\bf{\Lambda }} - {\bf{G}} + {\mathcal{C}}){(2{{\bf{\Lambda }}_c} - {\bf{G}} + {\mathcal{C}})^{ - 1}}\right]^{ - 1}} (2{{\bf{\Lambda }}_c} - {\bf{G}} + {\mathcal{C}})({\bf{\Lambda }} - {\bf{G}} + {\mathcal{C}}){\bf{z}}^\top\nonumber\\
&=& {\bf{z}}(2{{\bf{\Lambda }}_c} - {\bf{\Lambda }})\bigg\{ \frac{{{\bf{\Lambda }} + {\mathcal{C}}}}{2}(2{{\bf{\Lambda }}_c} + {\bf{\Lambda }} - 2{\bf{G}} + 2{\mathcal{C}}) \nonumber \\  &-& \big[{2{\bf{I}} - ({\bf{\Lambda }} - {\bf{G}} + {\mathcal{C}}){(2{{\bf{\Lambda }}_c} - {\bf{G}} + {\mathcal{C}})^{ - 1}}\big]^{ - 1}}(2{{\bf{\Lambda }}_c} - {\bf{G}} + {\mathcal{C}}) ({\bf{\Lambda }} - {\bf{G}} + {\mathcal{C}})\bigg\} {\bf{z}}^\top\nonumber\\
 &=& {\bf{z}}(2{{\bf{\Lambda }}_c} - {\bf{\Lambda }}){\big[2{\bf{I}} - ({\bf{\Lambda }} - {\bf{G}} + {\mathcal{C}}){(2{{\bf{\Lambda }}_c} - {\bf{G}} + {\mathcal{C}})^{ - 1}}\big]^{ - 1}}{(2{{\bf{\Lambda }}_c} - {\bf{G}} + {\mathcal{C}})^{ - 1}} \nonumber\\ &\times& \bigg\{ \big[2(2{{\bf{\Lambda }}_c} - {\bf{G}} + {\mathcal{C}}) - ({\bf{\Lambda }} - {\bf{G}} + {\mathcal{C}})\big]\frac{{{\bf{\Lambda }} + {\mathcal{C}}}}{2} (2{{\bf{\Lambda }}_c} + {\bf{\Lambda }} - 2{\bf{G}} + 2{\mathcal{C}}) \nonumber\\ &-& {(2{{\bf{\Lambda }}_c} - {\bf{G}} + {\mathcal{C}})^2}({\bf{\Lambda }} - {\bf{G}} + {\mathcal{C}})\bigg\} {\bf{z}}^\top\nonumber\\
 &=& {\bf{z}}(2{{\bf{\Lambda }}_c} - {\bf{\Lambda }}){\big[2{\bf{I}} - ({\bf{\Lambda }} - {\bf{G}} + {\mathcal{C}}){(2{{\bf{\Lambda }}_c} - {\bf{G}} + {\mathcal{C}})^{ - 1}}\big]^{ - 1}}{(2{{\bf{\Lambda }}_c} - {\bf{G}} + {\mathcal{C}})^{ - 1}} \nonumber\\ &\times& \bigg\{ \big[2(2{{\bf{\Lambda }}_c} - {\bf{G}} + {\mathcal{C}}) - ({\bf{\Lambda }} - {\bf{G}} + {\mathcal{C}})\big]\frac{{{\bf{\Lambda }} + {\mathcal{C}}}}{2} \left[\left(2{{\bf{\Lambda }}_c} - {\bf{G}} +{\mathcal{C}}\right) +\left({\bf{\Lambda }} - {\bf{G}} + {\mathcal{C}}\right)\right] \nonumber\\ &-& {(2{{\bf{\Lambda }}_c} - {\bf{G}} + {\mathcal{C}})^2}({\bf{\Lambda }} - {\bf{G}} + {\mathcal{C}})\bigg\} {\bf{z}}^\top\nonumber\\
 &=& {\bf{z}}(2{{\bf{\Lambda }}_c} - {\bf{\Lambda }}){\big[2{\bf{I}} - ({\bf{\Lambda }} - {\bf{G}} + {\mathcal{C}}){(2{{\bf{\Lambda }}_c} - {\bf{G}} + {\mathcal{C}})^{ - 1}}\big]^{ - 1}}{(2{{\bf{\Lambda }}_c} - {\bf{G}} + {\mathcal{C}})^{ - 1}} \nonumber\\ &\times& \bigg\{ (2{{\bf{\Lambda }}_c} - {\bf{G}} + {\mathcal{C}})^2{\bf{G}} + \frac{1}{2}({{\bf{\Lambda }}} - {\bf{G}} + {\mathcal{C}})^2  \left[ (2{{\bf{\Lambda }}_c} - {\bf{G}} + {\mathcal{C}}) - ({{\bf{\Lambda }}} - {\bf{G}} + {\mathcal{C}}) \right] \nonumber \\ &+&\frac{1}{2} ({{\bf{\Lambda }}} - {\bf{G}} + {\mathcal{C}}){\bf{G}}\left[(2{{\bf{\Lambda }}_c} - {\bf{G}} + {\mathcal{C}}) + ({{\bf{\Lambda }}} - {\bf{G}} + {\mathcal{C}})\right]\bigg\}{\bf{z}}^\top\nonumber\\
&=& {\bf{z}}(2{{\bf{\Lambda }}_c} - {\bf{\Lambda }})\underbrace {{{\big[2{\bf{I}} - ({\bf{\Lambda }} - {\bf{G}} + {\mathcal{C}}){{(2{{\bf{\Lambda }}_c} - {\bf{G}} + {\mathcal{C}})}^{ - 1}}\big]}^{ - 1}}}_{ \ge 0}{(2{{\bf{\Lambda }}_c} - {\bf{G}} + {\mathcal{C}})^{ - 1}} \nonumber\\ &\times& \bigg\{ \underbrace {{{(2{{\bf{\Lambda }}_c} - {\bf{G}} + {\mathcal{C}})}^2}{\bf{G}}}_{ > 0} + \frac{1}{2}\underbrace {({\bf{\Lambda }} - {\bf{G}} + {\mathcal{C}})}_{ > 0}\underbrace {({\bf{\Lambda }} + {\mathcal{C}})}_{ > 0}\underbrace {\big[(2{{\bf{\Lambda }}_c} - {\bf{G}} + {\mathcal{C}}) - ({\bf{\Lambda }} - {\bf{G}} + {\mathcal{C}})\big]}_{ > 0}\bigg\} {\bf{z}}^\top \nonumber\\ &\ge& 0.
\end{eqnarray}
The proof is then completed.
\end{proof}

\subsection{Proof of Proposition 3}
\begin{proof}
Since we have
\begin{equation}
{\Phi _i}(T) =  - \left( {{\Xi _i}(T) + c\sum\limits_{j \in N} {\frac{{{\Xi _j}(T)}}{{2{b_j}}} - } \sum\limits_{j \in N} {{g_{ij}}\frac{{{\Xi _j}(T)}}{{2{b_j}}}} } \right),
\end{equation}
we can conclude that
\begin{eqnarray}
\Phi (T) &=&  - \left( {\Xi (T) + {\cal C}{{\bf{\Lambda }}^{ - 1}}\Xi (T) - {\bf{G}}{{\bf{\Lambda }}^{ - 1}}\Xi (T)} \right) \nonumber \\
 &=&  - \Xi (T)\left( {{\bf I} - ({\bf{G}} - {\cal C}){{\bf{\Lambda }}^{ - 1}}} \right).
\end{eqnarray}
Thus, we have
\begin{eqnarray}
\Phi (T) &=&  - \left( {{\bf I} - ({\bf{G}} - {\cal C}){{\bf{\Lambda }}^{ - 1}}} \right){\left( {{\bf I} - \frac{{T - 1}}{{T + 1}}\left( {{\bf{G}} - {\cal C}} \right){{\bf{\Lambda }}^{ - 1}}} \right)^{ - 1}}{\bf{a}}\nonumber \\
 &=&  - \left( {{\bf I} - ({\bf{G}} - {\cal C}){{\bf{\Lambda }}^{ - 1}}} \right)\left( {{\bf I} + \sum\limits_{s = 1}^\infty  {\left( {{{\left( {\left( {T - 1} \right)\left( {{\bf{G}} - {\cal C}} \right)} \right)}^s}{{\left( {\frac{{{{\bf{\Lambda }}^{ - 1}}}}{{T + 1}}} \right)}^s}} \right)} } \right){\bf{a}}\nonumber \\
 &=&  - \Bigg( {\bf{I}} + \sum\limits_{s = 1}^\infty \left( {{{\left( {\frac{{{\bf{G}} - {\cal C}}}{{\bf{\Lambda }}}} \right)}^s}{{\left( {\frac{{T - 1}}{{T + 1}}} \right)}^s}} \right) - ({\bf{G}} - {\cal C}){{\bf{\Lambda }}^{ - 1}} \nonumber \\ &-& \left( {\frac{{T + 1}}{{T - 1}}} \right)\sum\limits_{s = 1}^\infty  {\left( {{{\left( {\frac{{{\bf{G}} - {\cal C}}}{{\bf{\Lambda }}}} \right)}^{s + 1}}{{\left( {\frac{{T - 1}}{{T + 1}}} \right)}^{s + 1}}} \right)} \Bigg){\bf{a}}\nonumber \\
 &=&  - \Bigg( {\bf{I}} + \sum\limits_{s = 1}^\infty \left( {{{\left( {\frac{{{\bf{G}} - {\cal C}}}{{\bf{\Lambda }}}} \right)}^s}{{\left( {\frac{{T - 1}}{{T + 1}}} \right)}^s}} \right) - ({\bf{G}} - {\cal C}){{\bf{\Lambda }}^{ - 1}} \nonumber \\ &-& \left( {\frac{{T + 1}}{{T - 1}}} \right)\sum\limits_{s = 2}^\infty  {\left( {{{\left( {\frac{{{\bf{G}} - {\cal C}}}{{\bf{\Lambda }}}} \right)}^s}{{\left( {\frac{{T - 1}}{{T + 1}}} \right)}^s}} \right)} \Bigg){\bf{a}}\nonumber \\
 &=&  - \Bigg( {\bf{I}} + \sum\limits_{s = 1}^\infty \left( {{{\left( {\frac{{{\bf{G}} - {\cal C}}}{{\bf{\Lambda }}}} \right)}^s}{{\left( {\frac{{T - 1}}{{T + 1}}} \right)}^s}} \right) - ({\bf{G}} - {\cal C}){{\bf{\Lambda }}^{ - 1}} \nonumber \\ &-& \left( {\frac{{T + 1}}{{T - 1}}} \right)\sum\limits_{s = 1}^\infty  {\left( {{{\left( {\frac{{{\bf{G}} - {\cal C}}}{{\bf{\Lambda }}}} \right)}^s}{{\left( {\frac{{T - 1}}{{T + 1}}} \right)}^s}} \right)}  + \frac{{{\bf{G}} - {\cal C}}}{{\bf{\Lambda }}} \Bigg){\bf{a}}\nonumber \\
 &=&  - \left( {{\bf{I}} - \frac{2}{{T - 1}}\sum\limits_{s = 1}^\infty  {\left( {{{\left( {\frac{{{\bf{G}} - {\cal C}}}{{\bf{\Lambda }}}} \right)}^s}{{\left( {\frac{{T - 1}}{{T + 1}}} \right)}^s}} \right)} } \right){\bf{a}}.
\end{eqnarray}
Under Assumption 1, we have
\begin{equation}
\frac{{{\bf{G}} - {\cal C}}}{{\bf{\Lambda }}} {\bf 1} < 1,
\end{equation}
and thus we have
\begin{eqnarray}
\Phi (T) &<&  - \left( {{\bf{I}} - \frac{2}{{T - 1}}\sum\limits_{s = 1}^\infty  {\left( {{{\left( 1 \right)}^s}{{\left( {\frac{{T - 1}}{{T + 1}}} \right)}^s}} \right)} } \right) a{\bf{1}}\nonumber \\
 &=&  - \left( {{\bf{I}} - {\bf{I}}\frac{2}{{T - 1}}\frac{{\frac{{T - 1}}{{T + 1}}}}{{1 - \frac{{T - 1}}{{T + 1}}}}} \right) a {\bf{1}}\nonumber \\
 &=& 0.
\end{eqnarray}
The proof is now completed.
\end{proof}
\bibliography{bibfile}
\end{document}